%% file: main.tex
\newif\ifcomment
\commenttrue
\commentfalse %

\newif\ifanon
\anontrue
\anonfalse %

\newif\ifpods
\podstrue
\podsfalse %

\ifpods
\ifanon
  \documentclass[acmsmall,review,anonymous]{acmart} 
\else
  \documentclass[acmsmall]{acmart} 
\fi 

\else
\documentclass[11pt]{article} 
\fi

\ifanon
    \author{Anonymous}
\else
    \author{
      Yotam Kenneth-Mordoch
      \qquad
      Shay Sapir
      \\ Weizmann Institute of Science
      \\ \texttt{\{yotam.kenneth,shay.sapir\}@weizmann.ac.il}
    }
\fi
\date{}
\input{paper.tex}

\title{On the Adversarial Robustness of Online Importance Sampling}

\begin{document}

\ifpods

\input{abstract}
\maketitle
\else
\maketitle
\input{abstract}
\fi

\input{introduction}

\section{Technical Overview}
\label{sec:technical-overview}
\input{technical-overview}

\section{Importance Sampling with Adversarial Sensitivities}
\label{sec:adversarial-sensitivities}

\input{importance_wth_adaptive_sensitivities}

\section{Application: Unweighted Hypergraph Cut Sparsification}
\label{sec:hypergraph-sparsification}
\input{hypergraph-sparsification}

\ifpods

\else
\section{Application: Subspace Embedding}
\label{sec:subspace-embedding}

\input{matrix_spectral_net_argument}
\fi

\ifanon
\else
\section{Acknowledgements}
We would like to thank Samson Zhou for pointing out the relevance of \cite{JPW23} for online sampling, and Sandeep Silwal for helpful comments on improving the readability of this work.
\fi

{\small
  \bibliographystyle{alphaurl}
  \bibliography{importance-robust}
} %

\appendix

\ifpods
\section{Application: Subspace Embedding}
\label{sec:subspace-embedding}

\input{matrix_spectral_net_argument}
\section{Deferred Hypergraph Proofs}
\label{sec:deferred-hypergraph}
\input{lemma-4.3}
\fi

\input{merge_and_reduce_robust_proof}

\end{document}

%% file: paper.tex
\ifpods
\else
\usepackage[letterpaper, margin=1in]{geometry}
\fi
\usepackage{hyperref}
\hypersetup{colorlinks=true,allcolors=blue}
\usepackage{amsthm}
\usepackage{amsmath}
\usepackage{algorithm}
\usepackage[noend]{algpseudocode}

\usepackage{amsfonts}
\usepackage{import} 
\usepackage{xifthen}
\usepackage{mathtools}
\usepackage{arydshln}%
\usepackage{dsfont}
\usepackage[capitalise]{cleveref} 
\usepackage[cal=esstix]{mathalfa}
\usepackage{thm-restate}

\usepackage{comment}
\ifdefined\AddToHook
\AddToHook{cmd/appendix/before}{\crefalias{section}{appendix}}
\AddToHook{cmd/appendix/before}{\crefalias{subsection}{appendix}}
\fi

\newcommand{\mintcut}{\mathrm{cut}}

\newcommand{\N}{\mathbb{N}}

\newcommand{\U}{\mathcal{U}}
\newcommand{\Q}{\mathsf{Q}}
\newcommand{\A}{\mathcal{A}}

\newcommand{\R}{\mathbb{R}}
\newcommand{\E}{\mathbb{E}}

\newcommand{\eps}{\epsilon}

\newcommand{\Norm}[1]{|| #1||}
\newcommand{\Mat}[1]{\mathbf{#1}}

\newcommand{\samp}{samp}
\newcommand{\post}{post}

\newcommand{\supp}{\operatorname{supp}}
\newcommand{\update}{\mu}

\newcommand{\adversary}{\operatorname{Adversary}}
\newcommand{\sampler}{\operatorname{SamplingAlg}}

\newcommand{\tO}{\tilde{O}}
\DeclareMathOperator{\poly}{poly}
\DeclareMathOperator{\polylog}{polylog}
\newcommand{\indic}[1]{ \mathds{1}_{\{#1\}} } %
\providecommand{\set}[1]{{\{#1\}}}
\newcommand{\eqdef}{\coloneqq}

\DeclareMathOperator{\var}{Var}
\DeclareMathOperator{\spn}{span}

\usepackage[colorinlistoftodos,textsize=tiny,textwidth=2cm,color=red!25!white,obeyFinal]{todonotes}

\newcommand{\colnote}[3]{\textcolor{#1}{$\ll$\textsf{#2}$\gg$\marginpar{\tiny\bf #3}}}

\ifcomment
\newcommand{\rnote}[1]{\colnote{purple}{#1--Robi}{RK}}
\newcommand{\ynote}[1]{\colnote{blue}{#1--Yotam}{YK}}
\newcommand{\ssnote}[1]{\colnote{orange}{Shay: #1}{SS}}
\else
\newcommand{\rnote}[1]{}
\newcommand{\ynote}[1]{}
\newcommand{\ssnote}[1]{}

\fi

\newtheorem{theorem}{Theorem}[section]

\newtheorem{claim}[theorem]{Claim}

\newtheorem{lemma}[theorem]{Lemma}

\newtheorem{definition}[theorem]{Definition}

\newtheorem{fact}[theorem]{Fact}
\newtheorem{remark}[theorem]{Remark}

%% file: abstract.tex
\begin{abstract}
    Online sampling algorithms, which irrevocably either keep or discard each stream element, 
    have seen wide use in streaming due to their efficiency and simplicity.
    Braverman et al. [NeurIPS 2021] claimed that \emph{online importance-sampling} algorithms, where elements are sampled proportionally to some notion of importance, succeed with high probability when their input stream is \emph{adaptively chosen by an adversary}.
    Unfortunately, their results on importance sampling do not beat trivial bounds in many instances.
    Therefore, we reopen the question about the robustness of online importance sampling to adaptive inputs.
    This question was also addressed by Jiang, Peng and Weinstein [FOCS 2023] for the problem of $\ell_2$-subspace embedding.
 
    We develop a unified framework for online importance sampling algorithms in adaptive streams. 
    This framework offers two main advantages: first, it provides better bounds than prior work, and second, it unifies and simplifies the analysis of importance sampling algorithms across different problems.
    We then leverage the framework to provide algorithms for cut sparsification in hypergraphs and $\ell_p$-subspace embeddings in adaptive streams whose space complexity nearly matches the oblivious case (non-adaptive).

\end{abstract}

%% file: introduction.tex
\section{Introduction}

The streaming model of computation is a rich algorithmic area, and particularly useful for large-scale data analysis.
A streaming algorithm is given its input as a sequence of items that can only be read sequentially,
and is required to compute some global function of the data. 
The main measure of a streaming algorithm's efficiency is its \emph{space complexity}, i.e., the amount of space it uses.
This model is particularly useful for the analysis of massive datasets, where the input is too large for the storage available to the algorithm.
This occurs naturally in many instances such as computing statistics of large databases, IoT measurements and network traffic logs. 
A more restricted variant of the streaming model is the \emph{online} model, where the algorithm may only store a small number of items (along perhaps with some auxiliary data structures) and its decisions are irrevocable, i.e., once an item is stored it may never be deleted.
While such online algorithms in general provide weaker guarantees, they are often simpler to analyze and implement. 
For further motivation, see e.g.~\cite{CMP16,BDM+MUWZ20}.

Most of the streaming literature assumes that the input stream is fixed in advance by an \emph{oblivious adversary}.
However, these assumptions do not necessarily hold, and a recent line of work \cite{Ben-EliezerY20,BJWY22,AlonBDMNY21,KMNS21_separating,WoodruffZ21,ABJSSWZ22,HassidimKMMS22,BEO22,CGS22,AssadiCGS23,AttiasCSS24,Stoeckl23,WoodruffZ24,ChakrabartiS24} considers the more difficult setting where the stream may depend on previous algorithm outputs, modeled by an \emph{adaptive adversary}.
That is, the algorithm must output a correct response after processing each item; 
the adversary may then observe these responses and choose the next stream element.
The immediate motivation is when the input is controlled by a malicious party, but another motivating scenario is when a user repeatedly queries and updates a database based on the answers to previous queries.
A streaming algorithm is called \emph{adversarially-robust} if it succeeds with high probability against any adaptive adversary.

This adaptivity introduces dependencies that break the analysis of most algorithms designed for oblivious streams.
Furthermore, there are several demonstrated adversarial attacks against classical algorithms, notably against linear sketches~\cite{HW13,BJWY22,CGS22,AhmadianEdith24,EdithLyuNelsonSSS22,CohenNelsonSarlosStemmer23,GribelyukLWYZ24,GribelyukLWYZ25,cohen2025attackrulealltight,ahmadian2025costcompressiontightquadratic}.
On the other hand, in some cases, adversarial-robustness is obtained with space complexity similar to the oblivious setting.
For example, adversarially-robust algorithms for frequency moments in insertion-only streams require at most a poly-logarithmic overhead~\cite{BJWY22,WoodruffZ21,HassidimKMMS22,AttiasCSS24}.

In this paper, we study the problem of constructing succinct representations of the stream.
This is a common preprocessing step in many algorithms where the data is too large to be processed in its entirety.
We study two seemingly unrelated problems in this vein: hypergraph cut sparsification and subspace embedding, both of which have seen wide use in streaming algorithms and beyond.
In both problems, the goal is to preserve some approximate property of the data while using as little storage as possible.
For hypergraphs, a natural generalization of graphs where hyperedges connect any number of vertices, we wish to find a small hypergraph that preserves all cuts up to a multiplicative $(1\pm \epsilon)$ factor \cite{GMT15,KK15,CX18,CKN21,KKTY21}.
Applications of cuts in hypergraphs include scientific computing on sparse matrices~\cite{BallardDKS16}, and clustering and machine learning\cite{ZhouHS06,YadatiNYNLT19}.
In subspace embedding, the input is a matrix $\Mat{A}\in \R^{n\times d}$ and the goal is to find a smaller matrix $\tilde{\Mat{A}}\in\R^{n'\times d}$ such that for every $x\in \R^d$ we have $\Norm{\tilde{\Mat{A}}x}\in (1\pm\epsilon)\Norm{\Mat{A}x}$ \cite{CohenP15,LL22,MuscoMWYasuda22,WoodruffYasuda23}.
Subspace embeddings have many applications in numerical linear algebra, see the surveys~\cite{Woodruff14,MartinssonT20}.
We solve both problems through the construction of coresets, weighted subsets of the original stream that preserve the desired property.
One useful approach for coreset construction in oblivious streaming is (online) importance sampling;
a weighted sampling technique where each stream element is assigned an importance and sampled with probability proportional to it (in the online analogue, items are assigned probabilities when they arrive and are sampled irrevocably).

The study of adversarially robust sampling was initiated by Ben-Eliezer and Yogev~\cite{Ben-EliezerY20}, who showed that uniformly sampling $O(\epsilon^{-2} \log|\U|)$ elements is both sufficient and necessary for $\epsilon$-estimation of the input stream, where $\U$ is the universe of stream elements (see also \cite{AlonBDMNY21} for better bounds parametrized by the Littlestone's dimension of the underlying set system).
Another line of work studies the robustness of online importance sampling algorithms \cite{BravermanHMSSZ21,JPW23}.
These results assume some ``condition'' bound $\kappa$ on the stream; for example, the ratio between the minimum and maximum cut in a graph.
The analysis of \cite{BravermanHMSSZ21} showed that using $\kappa^2$ space complexity overhead in comparison to oblivious algorithms results in adversarial-robustness.
Unfortunately, in graphs, the ratio between the minimum and maximum cuts is always at least the ratio between the minimum degree and half the number of edges in the graph, which is $\Omega(n)$.
Therefore, their algorithm requires storing the entire graph.
This overhead was improved for $\ell_2$-subspace embedding \cite{JPW23}.

Our work provides two main improvements to the construction of adversarially-robust online importance sampling algorithms.
First, we provide improved bounds, and second,
we provide a generic approach that unifies the analysis of importance sampling across many problems.
(In contrast to prior work, where each problem required its own problem-specific analysis.)
This unified approach simplifies the analysis and helps sheds some light on the difficulty of obtaining better bounds for online sampling algorithms.

Finally, by combining our online algorithms with a technique of \cite{Cohen-AddadWZ23,CWXZ25}, which integrates online sampling with the well-known merge-and-reduce framework, we obtain nearly-optimal hypergraph cut sparsification and $\ell_p$ subspace embedding.

\paragraph{Parallel work.}
The first version of this paper included only results on online sampling. 
Subsequently, the online posting of \cite{CWXZ25} inspired the addition of the technique combining online sampling with merge-and-reduce to further improve space complexity.
We note that our improved online-sampling algorithms are required for obtaining adversarially-robust streaming algorithms that match their oblivious counterparts.

\subsection{Hypergraph Cut Sparsification}\label{sec:intro_hypergraphs}
A hypergraph $G=(V,E)$ is a generalization of a graph, where edges (called hyperedges) can connect any number of vertices (i.e., every $e\in E$ is a subset of $V$).
One fundamental object in the study of hypergraphs is a cut, 
which is a partition of the vertex set $V$ into two disjoint sets $S\subseteq V$ and $V\setminus S$, and whose value is defined as $\mintcut_G(S) \eqdef \sum_{e\in E} \indic{0<|e\cap S| < |e|} \cdot w_e$.
Notably, the number of hyperedges can be as large as $2^{|V|}$, and
therefore, computing exact cuts in hypergraphs is often infeasible.
This motivates the constructions of succinct cut sparsifiers that preserve the cut values of the hypergraph.

\begin{definition}
    Given a hypergraph $G=(V,E)$, a reweighted subgraph $G'=(V,E')$ of $G$ is called a quality \emph{$(1\pm \epsilon)$-cut sparsifier} of $G$ if,
    \begin{equation*}
        \forall S\subseteq V,
        \qquad
        \mintcut_{G'}(S) \in (1\pm \epsilon) \cdot \mintcut_{G}(S)
        .
    \end{equation*}
\end{definition}
Hypergraph cut sparsifier construction is a well-studied problem \cite{KK15,CX18,BST19,CKN21,KKTY21,Quanrud2024}, including in the streaming setting \cite{GMT15,STY24,KPS24,KLP25,KPS25}.
We consider this problem in the insertion-only streaming model, where the hyperedges are given one at a time, and the stream's length is the number of edges denoted by $m$.
Our first result is an adversarially-robust algorithm for hypergraph cut sparsification.
Throughout we use $\tO(x)$ to hide polylogarithmic factors of $x$.
\begin{theorem}
    \label{theorem:hypergraph-cut-sparsifier-combined}
    Let $\epsilon>0$ and a vertex set $V$ of size $n$.
    There exists an algorithm that, given an adaptive stream of $m$ hyperedges $e_1,\ldots, e_m$ on $V$, maintains a $(1\pm \epsilon)$-cut sparsifier of $G_t=(V,\set{e_i}_{i=1}^t)$ for all $t\in [m]$.
    The algorithm succeeds with probability at least $1-1/\poly(n)$ and stores at most $\tO(\epsilon^{-2} n)$ hyperedges.
\end{theorem}
Our algorithm matches (up to $\polylog$ factors in $n,\log m,\log \epsilon^{-1}$) existing offline algorithms for hypergraph sparsification \cite{KK15,CX18,CKN21,Quanrud2024}.
It similarly matches existing algorithms for insertion-only (non-robust) streams \cite{GMT15,STY24,KPS24,KLP25,KPS25}.
Previously, \cite{BravermanHMSSZ21} obtained adversarial-robust algorithms via the merge-and-reduce framework, which is a general technique that applies to coresets in general. For hypergraphs, this yields
a robust sparsifier with 
$\tO(\epsilon^{-2}n\log^3 (m/n))$ hyperedges. 
Note that since $\log m$ can be as large as $n$ in hypergraphs, \Cref{theorem:hypergraph-cut-sparsifier-combined} offers substantial savings over existing adversarially-robust algorithms in the natural case when $m$ is large.
In addition, an adversarially robust hypergraph sparsification algorithm that stores $\tO(\epsilon^{-2}n)$ hyperedges was proposed in \cite{CWXZ25}.
However, their algorithm stores an auxiliary data-structure to compute sampling probabilities, which takes $\tO(\epsilon^{-2}n\poly(r))$ storage, where $r$ is the cardinality of the largest hyperedge in $H$.
Noting that $r$ can be as large as $\Omega(n)$ we find that the storage complexity of their algorithm is polynomially worse than ours.

Note that while the theorem is stated for unweighted hypergraphs, it can easily be extended to weighted hypergraphs by simulating the insertion of each hyperedge $e$ with weight $w_e$ as the insertion of $w_e$ copies of $e$.
This increases the storage requirement to $\tO(\epsilon^{-2} n \log\log W)$ hyperedges, where $W=\sum_{e\in E} w_e$ is the sum of all hyperedge weights.
We also give an improved algorithm for the online setting.
\begin{theorem}
    \label{theorem:hypergraph-cut-sparsifier}
    Let $\epsilon>0$ and a vertex set $V$ of size $n$.
    There exists an online sampling algorithm that, given an adaptive stream of $m$ hyperedges $e_1,\ldots, e_m$ on $V$, maintains a $(1\pm \epsilon)$-cut sparsifier of $G_t=(V,\set{e_i}_{i=1}^t)$ for all $t\in [m]$.
    The algorithm succeeds with probability at least $1-2^{-n}$ and stores at most $\tO(\epsilon^{-2} n^2 \log m)$ hyperedges.
\end{theorem}
Note that this algorithm has higher probability of success than \Cref{theorem:hypergraph-cut-sparsifier-combined}.
Previously, an adversarially-robust online sampling algorithm with $\tO(\kappa^2 \epsilon^{-2}n^2 \log m)$ edges for graphs was given in \cite{BravermanHMSSZ21}, where $\kappa$ is the ratio between the smallest and largest cut in the graph.
This result can be extended to hypergraphs using the bound on the number of hyperedges employed in the proof of \Cref{theorem:hypergraph-cut-sparsifier}.
Unfortunately, $\kappa$ can be $\Omega(2^n)$ and hence the algorithm does not improve upon the trivial solution of storing the entire stream in the worst case.
In addition, an adversarially-robust online sampling algorithm storing $\tO(\eps^{-2}n^2 \log^2 m)$ hyperedges can be obtained using the techniques of \cite{CWXZ25}, this algorithm has an $O(\log m)$ overhead factor in comparison to \Cref{theorem:hypergraph-cut-sparsifier}. 
Finally, note that in the online setting, there exists a lower bound of $\Omega(\epsilon^{-2} n \log m)$ on the number of hyperedges that must be stored in the construction of a cut sparsifier \cite{KLP25}.
Hence, the gap between our algorithm and the best possible result for online sampling (even in non-adaptive streams) is $\Theta(n)$.

\subsection{\texorpdfstring{$\ell_p$ Subspace Embedding}{lp Subspace Embedding}}\label{sec:intro_subspace_embedding}
We also consider a fundamental problem in numerical linear algebra, $\ell_p$ subspace embedding for $p>0$. 
In this problem, the input is a matrix $\Mat{A}\in\R^{n\times d}$ where $n\gg d$ and an accuracy parameter $\eps>0$, and the goal is to produce a (smaller) matrix $\tilde{\Mat{A}}\in \R^{n'\times d}$ such that $\|\tilde{\Mat{A}} x\|_p^p\in (1\pm \eps) \|\Mat{A}x\|_p^p$ for all $x\in\R^d$, where $\|y\|_p^p=\sum_{i=1}^n |y_i|^p$ for $y\in\R^n$.
A notable special case is $p=2$, also known as spectral approximation.
Oftentimes, it is desired that the rows of $\tilde{\Mat{A}}$ are a (weighted) subset of the rows of $\Mat{A}$, e.g., if the rows of $\Mat{A}$ are sparse then so are the rows of $\tilde{\Mat{A}}$.
Therefore, we restrict the output matrix $\tilde{\Mat{A}}$ to be constructed by a weighted subset of the rows of $\Mat{A}$.
In this setting, there are offline algorithms storing $\tO(\eps^{-2} d^{\max(1,p/2)})$ rows~\cite{CohenP15,MuscoMWYasuda22,WoodruffYasuda23}.

We consider the row-order streaming model, where the matrix is given row by row.
Denote by $\Mat{A}_i$ the matrix $\Mat{A}$ restricted to the first $i$ rows.
Define the \emph{online condition number} $\kappa^{OL}$ of $\Mat{A}$ to be the ratio between the largest singular value of the final matrix $\Mat{A}_n\equiv \Mat{A}$ and the smallest non-zero singular value across all intermediate matrices $\Mat{A}_i$.
We make the standard assumption that the entries of the matrix are integers bounded by $\poly(n)$ (so they can be stored in memory using $O(\log n)$ bits).

\begin{theorem}\label{thm:main_subspace_streaming}
    Let $p>0, \eps>0$ and $d\in \N$.
    There exists an algorithm that, given an adaptive stream of rows $a_1,\ldots, a_n\in \R^d$ whose entries are integers in $[-\poly(n),\poly(n)]$, maintains a $(1+\eps)$-approximate $\ell_p$ subspace embedding of $\Mat{A}_t=[a_1;\ldots;a_t]$ for all $t\in [n]$.
    The algorithm succeeds with high probability and stores at most $\tO\Big(\eps^{-2} d^{\max(1,p/2)} (\log\log (n\kappa^{OL}))^4\Big)$ rows.
\end{theorem}
This algorithm matches existing non-robust streaming and offline algorithms up to $\poly(\log d \cdot \log\log (n\kappa))$ factors~\cite{CohenP15,MuscoMWYasuda22,WoodruffYasuda23,CMP16,BDM+MUWZ20}.
Previously, there was an adversarially robust $\ell_p$-subspace embedding with an overhead of $d\kappa^{OL}$ factor~\cite{BravermanHMSSZ21}. This bound was improved for $p=2$ to $d\log\kappa^{OL}$ by~\cite{JPW23}.
Another approach by \cite{BravermanHMSSZ21} uses the merge-and-reduce framework, and has an overhead of $O(\log^3 n)$ over offline constructions.
\Cref{thm:main_subspace_streaming} beats merge-and-reduce only when $n\gg d$.
Similarly to hypergraph sparsification, we also provide an online sampling algorithm.

\begin{theorem}\label{thm:application_ellp_subspace}

    Let $p>0, \eps,\delta>0$ and $d\in \N$.
    There exists an online algorithm that, given an adaptive stream of rows $a_1,\ldots, a_n\in \R^d$ whose entries are integers in $[-\poly(n),\poly(n)]$, maintains a $(1+\eps)$-approximate $\ell_p$ subspace embedding of $\Mat{A}_t=[a_1;\ldots;a_t]$ for all $t\in [m]$.
    The algorithm succeeds with probability $1-\delta$ and stores at most $O\Big(\eps^{-2} (d\log\tfrac{\kappa^{OL}}{\eps} + \log\log n+\log\tfrac{1}{\delta})\cdot(d\log (n\kappa^{OL}))^{\max(1,p/2)}\Big)$ rows.
\end{theorem}
This result matches the current best-known adversarially-robust \emph{online} algorithms in row-order streams for $p=2$ \cite{JPW23}.
We extend it to all $p>0$, which is straightforward given our framework.
There remains a gap of roughly $O(d\log \kappa^{OL})$ to the known online algorithms in the non-adaptive setting (suppressing logarithmic factors),
which store
$\tO(\eps^{-2} (d\log (n\kappa^{OL}))^{\max(1,p/2)})$ rows~\cite{WoodruffYasuda23}.

\subsection{Organization}
The rest of the paper is organized as follows.
\Cref{sec:technical-overview} provides an overview for the techniques used in the proofs, then in \Cref{sec:adversarial-sensitivities} we provide the proofs for our self-weighted online sampling framework.
Finally, in \Cref{sec:hypergraph-sparsification,sec:subspace-embedding} we give the details for hypergaph cut sparsification and subspace embedding, respectively.

%% file: technical-overview.tex
Our adversarially-robust algorithms are built upon the following scheme.
We first give a framework for adversarially-robust online importance sampling algorithms that choose sampling probabilities based solely on items sampled so far, which we term \emph{self-weighted}.
We then combine these algorithms with the widely applicable merge-and-reduce framework.

Our online sampling result provides a clean and generic approach for adversarially-robust self-weighted sampling algorithms by unifying the approach of \cite{BravermanHMSSZ21,JPW23} to make it easily applicable to any self-weighted online sampling algorithm.
Furthermore, our approach improves the parameters of their constructions. 
The proof is based on one-dimensional importance sampling, which we show is inherently adversarially robust. 
We extend this result by using a union bound on all ''dimensions'', as we discuss at the end of the next section.

Our combination of  adversarially-robust self-weighted online sampling with the merge-and-reduce framework is formalized as a black-box wrapper.
This technique was used before to improve the storage complexity of \emph{oblivious} streams \cite{Cohen-AddadWZ23,CWXZ25}.
(Which \cite{CWXZ25} then use as a basis for an adversarially-robust algorithm using the computational-paths framework.)
The basic idea is that the output sequence of an online algorithm can be fed, without storing it, as a virtual input stream to a merge-and-reduce algorithm. 
The sampling probabilities for the online sampling algorithm are then obtained based on the output of the merge-and-reduce algorithm.
The correctness argument for adversarial-robustness is a bit delicate, but it essentially follows from the adversarial-robustness of merge-and-reduce~\cite{BravermanHMSSZ21}, and the fact that our online sampling algorithms are self-weighted.

\subsection{Self-Weighted Online Sampling Framework}
Our online sampling framework is based on showing that the one-dimensional case, defined as follows, is adversarially robust.
\begin{definition}[One Dimensional Self-Weighted Online Importance Sampling]\label{def:online-importance-sampling}
    Given an input stream $x_1,\ldots,x_m\in \R_+$, self-weighted online importance-sampling with amplification parameter $a>1$ is the following algorithm.
    Upon receiving item $x_t$, set $1\geq p_t\geq \min\{1,a \tfrac{x_t}{x_t+\sum_{i=1}^{t-1}\tilde{x}_i}\}$, and 
    use fresh randomness to compute
    \[
    \tilde{x}_t =
    \begin{cases}
        \tfrac{x_t}{p_t} \qquad \text{w.p. $p_t$}, \\
        0 \qquad \text{otherwise.}
    \end{cases}
    \]
    For every $t\leq m$, return $\sum_{i=1}^t \tilde{x}_i$ as an estimate for $\sum_{i=1}^t x_i$.
\end{definition}
We say that $\tilde{y}$ is a $(1+\eps)$-approximation of $y$ if $\tilde{y}\in (1\pm\eps)\cdot y$.
The adversarial robustness of importance sampling was first examined in~\cite{BravermanHMSSZ21}, who showed that
given a deterministic (but crude) bound $\Delta>1$ on the input,
which is roughly the sum of elements in the stream in the worst-case dimension,
one can get an adversarially-robustness by paying
a $\poly(\Delta)$ factor in the storage complexity compared to the non-adaptive setting.
This result was improved for $\ell_2$ subspace embedding, to a factor of roughly $\log \Delta$~\cite{JPW23}.

Our approach extends the techniques of \cite{JPW23} to all self-weighted online sampling problems. 
Additionally, using an easy observation, 
we improve the ``cost'' of adversarial robustness from $\log \Delta$ to $\log \log \Delta$.%
\footnote{Ignoring factors depending on $\epsilon^{-1}$.}
Finally, this framework is widely applicable, as is demonstrated by our two applications, hypergraph cut sparsification and subspace embedding.
\begin{theorem}[Adversarially-Robust Self-Weighted Importance Sampling (Correctness)]\label{thm:importance_samp_adv_robust_main}
    Let $\epsilon, \delta \in (0,1)$, $\Delta > 1$.
    Given an adaptive stream of non-negative numbers $x_1,\ldots,x_m\in \R_+$ such that $\frac{\sum_{i=1}^m x_i}{x_1} \leq \Delta$;
    with probability at least $1-\delta$, self-weighted online importance-sampling with amplification parameter $a = O(\epsilon^{-2} \log \frac{\log \Delta}{\epsilon \delta})$ returns a $(1+\epsilon)$-approximation of $\sum_{i=1}^t x_i$ for all $t\in [m]$.
\end{theorem}
Note that our theorem focuses on bounding the amplification parameter $a$ and not the actual sample size.
We note that factor $\log\log \Delta$ seems necessary also for algorithms in oblivious streams that guarantee correctness at every time step.

The assumption $\tfrac{\sum_{i=1}^m x_i}{x_1} \leq  \Delta$ can be replaced with the natural (and stronger) assumption that the updates are bounded in $[1,\Delta']$, which yields $\Delta \le m\Delta'$.
Moreover, some bound on update size must be assumed, since otherwise, the sum of online importances $\sum_{t=1}^{m} \tfrac{x_t}{x_t+\sum_{i=1}^{t-1}\tilde{x}_i}$ may be as large as $\Omega(m)$, and the algorithm must then store the entire stream.
For example, consider the stream $1,2,4,\ldots, 2^m$ with amplification parameter $a=O(1)$. 
At time $t\in[m]$, we have $\tfrac{2^t}{\sum_{i=0}^{t} 2^i}=\Omega(1)$, hence $p_t=1$, and eventually all items are sampled.
We now give an overview of the proof of \Cref{thm:importance_samp_adv_robust_main}.

\paragraph{The adversary's power.}
Recall that every item is irrevocably kept with probability proportional to its importance at the moment it arrives.
Therefore, once an item is processed by the algorithm, the adversary cannot affect it anymore.
Hence, the adversary can only hope to ``fail'' the algorithm by either changing the sampling probabilities or by adding ``bad'' items to the stream.

Our proof follows by separating the adversary's power into these parts: inserting items and setting sampling probabilities.
We first show that if the sampling probabilities are ``good'', then the algorithm maintains an accurate estimate with high probability (for the amplification parameter $a$ of \Cref{thm:importance_samp_adv_robust_main}).
We then show through a bootstrapping argument that the sampling probabilities are indeed ``good'' with high probability. %

\paragraph{Sampling game.}
For the first part, consider a two-player game between a sampling algorithm, $\sampler$, and an adversary $\adversary$. 
In this game, the adversary essentially has more power compared to \Cref{thm:importance_samp_adv_robust_main} --- the adversary also picks the  sampling probabilities subject to some constraint.
The game is as follows.
Let $\epsilon\in (0,1)$.
First, $\sampler$ picks a number $a\geq 1$.
Then the game proceeds in rounds, where in the $t$-th round,
\begin{enumerate}
    \item $\adversary$ picks a number $x_t>0$, and assigns it a sampling probability $\min\{a\tfrac{x_t}{\sum_{i=1}^t x_i},1\} \leq p_t\leq 1$, and sends $(x_t,p_t)$ to $\sampler$.
    \item $\sampler$ uses fresh randomness and computes 
    \[
    \tilde{x}_t =
    \begin{cases}
        \tfrac{x_t}{p_t} \qquad \text{w.p. $p_t$}, \\
        0 \qquad \text{otherwise,}
    \end{cases}
    \]
    and sends $\tilde{x}_t$ to $\adversary$.
\end{enumerate}
The goal of $\sampler$ is to maintain $\sum_{i=1}^t \tilde{x}_i \in (1\pm\epsilon) \sum_{i=1}^t x_i$ for all $t$, and the goal of $\adversary$ is to cause $\sampler$ to return an incorrect estimate at some time $t$.
Notice that this game is similar to \Cref{def:online-importance-sampling}, but now the adversary has to use sampling probabilities $p_t$ that are constrained by the exact quantity $\sum_{i=1}^t x_i$, rather than its approximation $x_t + \sum_{i=1}^{t-1} \tilde{x}_i$.
The following technical lemma states that for the amplification parameter $a$ of \Cref{thm:importance_samp_adv_robust_main}, $\adversary$ loses the game with high probability.
\begin{lemma}[Sampling Game]\label{lemma:sampling-game-fixed-probabilities}
    Let $\Delta > 1, \epsilon,\delta\in(0,1)$.
    Consider the game between $\adversary$ and $\sampler$ with the restriction that 
    $\tfrac{\sum_{i=1}^m x_i}{x_1} \leq  \Delta$.
    For a suitable $a=O(\epsilon^{-2}\log \tfrac{\log \Delta}{\epsilon\delta})$, $\sampler$ wins the game with probability $1-\delta$.
\end{lemma}
In the oblivious (non-adaptive) setting, one can prove a similar lemma, essentially by a Bernstein's bound and by observing that the variance of $\sum_{i=1}^t x_i-\tilde{x}_i$ is bounded by $\tfrac{1}{a}(\sum_{i=1}^t x_i)^2$.
One might wish to use a similar method for the one-dimensional case in the adaptive setting, by defining an appropriate martingale sequence $X_t=\sum_{i=1}^t x_i-\tilde{x}_i$, and applying Freedman's inequality (which is analogous to Bernstein's inequality).
Unfortunately, in order to apply Freedman's inequality, we need a bound on $\sum_{i=1}^t x_i$,
which is a random variable in the adaptive setting.
We overcome this challenge by
partitioning the stream into $O(\epsilon^{-1} \log \Delta)$ phases, based on rounding $\sum_{i=1}^t x_i$ to the nearest power of $(1+\epsilon)$.
Note that this partition is used only for the analysis.

For each phase $k$, we create a virtual stream of items $\set{x_i'}_{i\in[m]}$, such that $x_j'=x_j$ for all $j\in[m]$ such that $\sum_{i=1}^j x_i \le (1+\epsilon)^k \cdot x_1$, and otherwise $x_j'=0$.
This yields a deterministic bound of $\sum_{i=1}^t x_i \le (1+\epsilon)^k \cdot x_1$ in the $k$-th stream.
We then define an appropriate martingale sequence for each virtual stream,
and use this deterministic bound on $\sum_{i=1}^t x_i$ to bound the martingale sequence using Freedman's inequality. 
The proof is concluded by applying a union bound over all virtual streams.

Note that previous work employed a similar technique, however it was only applied to the problem of $\ell_2$ subspace embedding, and furthermore it used $O(\eps^{-1}\Delta)$ phases, instead of $O(\eps^{-1}\log \Delta)$ phases as in our case \cite{JPW23}.
Thus, they require the amplification parameter $a$ to be $O(\log (\epsilon^{-1} \Delta))$ compared to the $O(\log (\epsilon^{-1} \log \Delta))$ that we achieve.
Finally, our proof technique has two main advantages.
First, by abstracting the problem to a game we obtain a simpler proof, which is easier to follow and extend.
Second, it enables the application of the same technique to other problems.
Therefore, we hope that this presentation will be a useful basis for future works.
For further details, see \Cref{sec:adversarial-sensitivities}.

\paragraph{Bootstrapping the sampling probabilities.}
We now explain how to strengthen \Cref{lemma:sampling-game-fixed-probabilities} to the case when the sampling probabilities are not computed deterministically, thus proving \Cref{thm:importance_samp_adv_robust_main}.
This follows by formalizing online importance sampling, as a version of the game between $\adversary$ and $\sampler$.
In this version, $a=O(\epsilon^{-2}\log \tfrac{\log (\Delta)}{\epsilon\delta})$ as in \Cref{lemma:sampling-game-fixed-probabilities}, and
$\adversary$ is required to choose $p_t = \min \left\{ \frac{2a x_t}{x_t+\sum_{i=1}^{t-1} \tilde{x}_i},1 \right\}$ (i.e., the ``online'' importance of $x_t$) whenever it is a valid strategy.
When this strategy is not valid, the adversary is not restricted.
Notice that if $\sampler$'s output was correct up to time $t$, then the above is indeed a valid strategy for the game, i.e., if
$\sum_{i=1}^{t-1} \tilde{x}_i\in (1\pm \epsilon) \sum_{i=1}^{t-1} x_i$, then
\begin{equation}
    \label{eq:modified-game-strategy}
    \frac{2a x_t}{x_t+\sum_{i=1}^{t-1} \tilde{x}_i}
    \geq \frac{2a x_t}{x_t+(1+\epsilon)\sum_{i=1}^{t-1} {x}_i}
    \geq \frac{a x_t}{\sum_{i=1}^{t} {x}_i},
\end{equation}    
and the strategy is valid.
\begin{proof}[Proof of \Cref{thm:importance_samp_adv_robust_main}]
    We consider a dominant strategy for an adversary that tries to fool online importance sampling.
    For every $t\in [m]$, the adversary picks some $\chi_t$ of their choice that satisfies $(\chi_t +\sum_{i\le t} x_i) /x_1 \le \Delta$, which can depend on past randomness.
    If $\frac{2a \chi_t}{\chi_t+\sum_{i=1}^{t-1} \tilde{x}_i}\geq \frac{a \chi_t}{\chi_t+\sum_{i=1}^{t-1} {x}_i}$, the adversary chooses $x_t=\chi_t$, and otherwise, they choose $x_t=0$.
    This is a dominant strategy, since the adversary can choose a strategy freely while $\sum_{i=1}^{t} \tilde{x}_i\in (1\pm \epsilon) \sum_{i=1}^{t} x_i$, and when this condition is violated, the future choices of the adversary do not affect the outcome (adversary had already won).

    Additionally, the strategy described above, along with the ``online importance'' of $\chi_t$, $p_t = \min \left\{ \frac{2a \chi_t}{\chi_t+\sum_{i=1}^{t-1} \tilde{x}_i},1 \right\}$, is a valid strategy for \Cref{lemma:sampling-game-fixed-probabilities}. (The factor $2$ can be incorporated in the parameter $a$.)
    Therefore, such an adversary loses with probability at least $1-\delta$, and since their strategy is dominant, this concludes the proof.
\end{proof}

\paragraph{Beyond the one-dimensional case.}
To obtain our bounds for hypergraph sparsification and subspace embedding, we apply \Cref{thm:importance_samp_adv_robust_main} with a union bound that we call ``uniform''. Before going into specific details, we define a general setting for which our approach is applicable, and is captured by the notion of coresets.
We consider problems defined by a universe $\U$ (e.g. $\U=\R^d$ for points, or $\U=2^{[n]}$ for hyperedges), query set $\Q$ (e.g. $\R^d$ for subspace embedding, or cuts for hypergraphs), and cost function $c:\U\times \Q \to \R_+$, where $\R+$ is the set of all non-negative reals.
For every query $x\in\Q$ and weight vector $w\in \R_+^{|\U|}$, where $R_+^{|\U|}$ is the set of all vectors in $\R^{|\U|}$ whose entries are non-negative, we define its cost in regards to $x$ as 
\[
C(w,x) \eqdef \sum_{u\in \U} w_u \cdot c(u,x).
\]
We denote the stream as $\set{(u_i,a_i)}_{i=1}^m$ where $u_i\in\U$ is an element and $a_i\in \R_+$ is its weight.
Define $\update_i\eqdef a_i\cdot e_{u_i}$ where $e_{u_i}$ is the unit vector in the direction $u_i$.
Using this, we can represent the input stream as $(\update_1,\ldots,\update_m)$, and define the vector of all inputs up to time $t$ as $w_t = \sum_{i\le t}\update_i$. 
(For example, think of $w_t$ as a weighted hypergraph represented in a vector form.)
Similarly, denote the output stream of the algorithm by $w'_t$.
We say the stream is adaptive if for every $t\leq m-1$, $\update_{t}$ may depend on $\{w'_1,\ldots,w'_{t-1}\}$.
A coreset is defined as follows.
\begin{definition}[Coreset]
    \label{def:coreset}
    For a cost function $C$, a $(1+\eps)$-coreset of $w$ is a vector $w'\in \R_+^{\U}$ such that $\supp(w')\subseteq \supp(w)$ and
    \[
    \forall x\in\Q, \qquad C(w,x)
    \in (1\pm\eps)\cdot C(w',x).
    \]
    The size of a coreset $w'$ is the number of non-zero coordinates it has.
    Coresets also satisfy the following properties.
    \begin{itemize}
        \item \textbf{Reduce:} If $w'$ is an $(1+\eps_1)$-coreset of $w$ and $w''$ is an $(1+\eps_2)$-coreset of $w'$, then $w''$ is an $(1+\eps_1)(1+\eps_2)$-coreset of $w$. 
        \item \textbf{Merge:} If $w_1'$ is an $(1+\eps_1)$-coreset of $w_1$ and $w_2$ is an $(1+\eps_2)$-coreset of $w_2'$, then $w_1'+w_2'$
        is a $(1+\max\{\eps_1,\eps_2\})$-coreset of $w_1+w_2$.
    \end{itemize}
\end{definition}
\begin{remark}
    Note that both hypergraph cut sparsifiers and $\ell_p$ subspace embeddings are in fact coresets, this is formally proven in \Cref{sec:hypergraph-coreset-properties} and \Cref{sec:subspace-coreset-properties} respectively.
\end{remark}

We say that a streaming algorithm computes/maintains a coreset if for every $t\leq m$, it outputs a vector $w'_t\in \R_+^{\U}$ that is a coreset of $w_t$.
Our framework crucially builds on online sampling algorithms, see e.g. \cite{AG09,CMP16,STY24,KLP25}, and we use a restricted notion that we call \emph{self-weighted}, as follows.
\begin{definition}[Self-Weighted Online Importance Sampling]
    \label{def:self-weighted-online}
    Given an input stream $\update_1,\ldots,\update_m$, self-weighted online importance sampling with amplification parameter $a>1$ is the following algorithm.
    Upon receiving item $\update_t$, set $1\geq p_t\ge \min \{1,a\cdot\max_{x\in Q} \frac{C(\update_t,x)}{C(w'_{t-1}+\update_t,x)} \}$ and use fresh randomness to compute
    \[
        \update'_t \gets 
        \begin{cases}
            \frac{\update_t}{p_t} & \text{with probability $p_t$, and}\\
            0 & \text{otherwise.}
        \end{cases}
    \]
    Maintain $w'_t = \sum_{i=1}^t \update'_i$.
\end{definition}

\paragraph{Putting It All Together.}
In our applications, we compute a coreset as follows.
Assume we are given some net $\Q'\subseteq \Q$ of bounded size, such that if $C(w',x)\in (1\pm \epsilon)C(w,x)$ for every $x\in \Q'$, then $C(w',y)\in (1\pm O(\epsilon))C(w,y)$ for every $y\in \Q$.
Fix some $x\in \Q'$ and observe that $C(w,x)$ is exactly a sum of elements as described in \Cref{thm:importance_samp_adv_robust_main}.
If we sample the element at time $t$ with probability at least $p_t= a \frac{C(\update_t,x)}{C(w'_t+\update_t,x)}$ (for brevity, we omit the minimum with $1$), then by \Cref{thm:importance_samp_adv_robust_main}, we obtain a $(1+\epsilon)$-approximation of $C(w,x)$, for an appropriate $a>1$.

To approximate $C(w,x)$ for all $x\in \Q'$, we sample the element at time $t$ with probability $a\cdot \max_{x\in \Q'}  \frac{C(\update_t,x)}{C(w'_t+\update_t,x)}$.
Next, we amplify the success probability of \Cref{thm:importance_samp_adv_robust_main} by a $\tfrac{1}{|\Q'|^2}$ (which increases space by a $\log |\Q'|$ factor).
By a union bound, we obtain correctness for all $x\in\Q'$, and hence $w_t'$ is a coreset of $w_t$.
We call such a union bound ``uniform'' because it is based on the same net for all $w\in \R_+^{\U}$.

To make this idea concrete, we now give an overview of the construction of a cut sparsifier for hypergraphs (\Cref{theorem:hypergraph-cut-sparsifier}).
The details for $\ell_p$ subspace embedding are similar and omitted for brevity. 
Throughout, let $G=(V,E)$ be some hypergraph.
To construct a cut sparsifier of $G$, we choose the net $\Q'$ to be the entire set of cuts $2^{[n]}$.
It is clear that $C(w,x)$ is simply the number of hyperedges crossing the cut for every $x\in 2^{[n]}$.
Therefore, the sampling probability of each hyperedge is given by the smallest cut which intersects it, and then applying a union bound on all the cuts to conclude the proof.%
\footnote{For simplicity, the algorithm in \Cref{theorem:hypergraph-cut-sparsifier} samples according to strong connectivity, which can be shown to give a lower bound on the size of the minimum cut intersecting $e$.}
The bound on the size of the sparsifier follows from structural analysis akin to \cite{AG09} and is detailed in \Cref{sec:hypergraph-sparsification}.

\paragraph{Gap to Oblivious Online Sampling.}
Unfortunately, the ``uniform'' union bound approach leaves a sizable bound to the oblivious setting.
For example, for $\ell_p$ subspace embeddings, in the oblivious setting, one can directly analyze the supremum of a certain quantity over the set $\{x:\|\Mat{A}x\|_{p} = 1\}$ using a standard symmetrization argument and some other clever arguments.
In comparison, our approach requires a uniform high probability bound for each element in the net $\Q$.
It is unclear how to employ such symmetrization arguments  in the adaptive setting, since the set  $\{x:\|\Mat{A}x\|_{p} = 1\}$ is now a random variable.
Hence, it remains open to close the gap between \Cref{thm:application_ellp_subspace,theorem:hypergraph-cut-sparsifier} and the oblivious setting for online sampling algorithms.

\subsection{Black-Box Wrapper: Online Sampling and Merge-and-Reduce}\label{sec:technical_merge_reduce}
We now present a black-box wrapper, based on a framework of \cite{Cohen-AddadWZ23,CWXZ25}, that takes a self-weighted online sampling algorithm and produces an algorithm with smaller space complexity (though no longer an online algorithm). 
The wrapper feeds the output of a self-weighted sampling algorithm to
a merge-and-reduce algorithm, and uses the output of the merge-and-reduce to compute the sampling probabilities for the former.
We show that if the online sampling algorithm is adversarially-robust, then so is the output of the combined algorithm.

\paragraph{Merge-and-reduce.}
The well-known merge-and-reduce framework is as follows.
First, assume there exists an offline algorithm that constructs a $(1+\eps)$-coreset of size $K(\eps)$, which we will use with $\eps'$ that is set later.
Next, partition the input stream into chunks of size $K=K(\eps')$. We construct a binary tree whose leaves are these chunks, and every node holds at most $K$ elements.
Whenever a node has that its two children store $K$ elements, it merges them to size $2K$ and then reduces the merged set back to size $K$ using the offline algorithm. We then clear the storage of the two children.
It is easy to see that the number of levels in this procedure is $\log (m/K)$, and that we store at most $2K$ elements in each level at the same time.

Observe that when a node collects and merges the elements of its two children, it obtains a $(1+\epsilon')^2$-coreset of their union by \Cref{def:coreset}.
Thus, after applying this procedure for $\log (m/K)$ levels, we obtain that the root holds a $(1+\epsilon')^{\log (m/K)}\leq 1+\eps$ coreset of the entire stream, where we set $\epsilon'=\tfrac{\eps}{3\log (m/K)}$.

Previously, Braverman et al.~\cite{BravermanHMSSZ21} claimed that merge-and-reduce is adversarially-robust, and we provide a short proof in \Cref{sec:merge_reduce_proof} for completeness.
\begin{theorem}[Adversarial-robustness of merge-and-reduce]\label{thm:merge_reduce_robust}
    Let a universe $\U$, a query set $Q$, a cost function $C$
    and $0<\eps<\tfrac{1}{2}, 0<\delta<1$. 
    Assume that for all $0<\eps'<\tfrac{1}{2}, 0<\delta'<1$, there exists an offline algorithm $\A_{\eps',\delta'}$ such that when it is given a vector $w\in\R_+^{\U}$, with probability $1-\delta'$, the algorithm outputs a $(1+\eps')$-coreset of $w$ of size $g(\eps',\delta')\geq 2$, for some function $g$.

    Then, there exists an \emph{adversarially-robust} streaming algorithm, 
    that given an input stream $(\update_1,\ldots,\update_m)$ in an adaptive stream, maintains a $(1+\eps)$-coreset of $P$ of size $O(g(\tfrac{\epsilon}{3\log m},\tfrac{\delta}{m})\cdot\log m)$.
    The algorithm succeeds with probability $1-\delta$.
\end{theorem}
\begin{remark}
    Suppose that every element in $\U$, and the weight of every coreset element, can be represented using $s$ bits of space.
    Then, the merge-and-reduce algorithm uses $O(s\cdot g(\tfrac{\epsilon}{3\log m},\tfrac{\delta}{m})\cdot\log m)$ bits of space.
\end{remark}

\paragraph{Black-box wrapper.}
Online sampling composes well with merge-and-reduce.
This idea, proposed by \cite{Cohen-AddadWZ23,CWXZ25}, is to apply a self-weighted online sampling algorithm on the input stream, and then feed its output stream into a merge-and-reduce procedure.
In addition, to avoid storing the output of sampling algorithm, we modify it to choose sampling probabilities according to the coreset constructed by the merge-and-reduce procedure.
Observe that this allows us to obtain much better storage complexity than directly using merge-and-reduce as the stream length is now much shorter.
An illustration of this process is provided in \Cref{fig:blackbox-reduction}.
The following theorem states the guarantees of our wrapper, and particularly, its relation to adversarial-robustness, which was not studied before.
\begin{theorem}
    \label{theorem:black-box-reduction}
    Let $\eps_1,\eps_2,\delta_1,\delta_2>0$.
    Suppose there exists an adversarially-robust self-weighted sampling algorithm $\A_{\samp}$ with amplification parameter $a>1$, that with probability $1-\delta$, constructs a $(1+\eps_1)$-coreset of size $h(m,\eps_1,\delta_1)$, where $m$ is the stream's length.
    Furthermore, suppose there exists an offline algorithm that computes with probability $1-\delta_1$ a $(1+\eps_2)$-coreset of size $g(\epsilon_2,\delta_2)$.

    Then, for all $\eps,\delta>0$, there exists an adversarially-robust algorithm that, given an adaptive stream of length $m$, with probability $1-\delta$, outputs a $(1+\eps)$-coreset of size $O\big(g(\tfrac{\eps}{3\log h(m,\frac{\eps}{3},\frac{\delta}{2})},\tfrac{\delta}{2h(m,\frac{\eps}{3},\frac{\delta}{2})})\cdot\log (h(m,\tfrac{\eps}{3},\frac{\delta}{2}))\big)$.
\end{theorem}
\begin{figure}[t]
\centering
\includegraphics{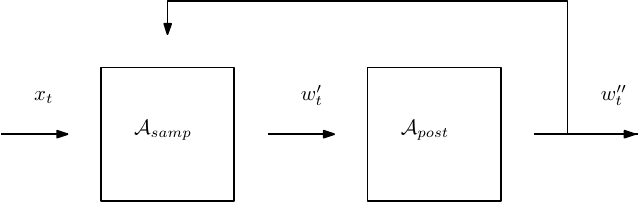}%
\caption{Black-box wrapper: the self-weighted sampler produces a compressed stream that is fed to merge-and-reduce, while merge-and-reduce supplies the coreset used to compute sampling probabilities.}
\label{fig:blackbox-reduction}
\end{figure}

\begin{proof}[Proof of \Cref{theorem:black-box-reduction}]
    We show that the algorithm described above, of combining self-weighted online sampling and merge-and-reduce, satisfies the guarantees of the theorem.
    Denote the output of $\A_{\samp}$ by $w'_t$, and the output of $\A_{\post}$ by $w''_t$.
    We will use the following claim.
    \begin{claim}\label{lem:online_and_merge_reduction}
        For every adaptive input stream $\update_1,\ldots,\update_m\in \U$, if $\A_{\samp}$ samples at time $t$ according to probability $1\geq p_t \geq (1+\epsilon_2)a\cdot \max_{x\in Q} \frac{C(\update_t,x)}{C(w_t''+\update_t,x)}$, then with probability $1-\delta_1-\delta_2$, the ouput $w''_t$ is a $(1+\epsilon_1)(1+\epsilon_2)$-coreset of $w_t$ for all $t\leq m$.
    \end{claim}
    \begin{proof}[Proof of \Cref{lem:online_and_merge_reduction}]
        By \Cref{thm:merge_reduce_robust}, we have that with probability $1-\delta_2$, for all $t\leq m$, the output $w_t''$ is a $(1+\eps_2)$-coreset of $w_t'$ (this holds for adaptive streams and hence for the stream $w'_t$).
        Therefore, we have $p_t = (1+\eps_2)a \cdot \max_{x\in Q} \frac{C(\update_t,x)}{C(w_t''+\update_t,x)} \ge a\cdot\max_{x\in Q} \frac{C(\update_t,x)}{C(w_t'+\update_t,x)}$. 
        Hence, the self-weighted sampling algorithm succeeds with probability $1-\delta_1$.
        By the law of total probability, we obtain that both subroutines succeed with probability $1-\delta_1-\delta_2$, and the proof is concluded by the Reduce property of \Cref{def:coreset}.
    \end{proof}
    Using the claim we immediately obtain the correctness guarantee of  \Cref{theorem:black-box-reduction}.
    The bound on the size is using \Cref{thm:merge_reduce_robust} by observing that the virtual stream inserted to the merge-and-reduce algorithm is of length $h(m,\tfrac{\eps}{3},\delta_1)$.
\end{proof}

\paragraph{Hypergraph Cut Sparsification.}
We continue with the running example of hypergraph cut sparsification, and employ the reduction using the self-weighted online sampling algorithm of \Cref{theorem:hypergraph-cut-sparsifier}.
The proof for subspace embeddings (\Cref{thm:main_subspace_streaming}) follows similar lines, and is deferred to \Cref{sec:subspace_merge_reduce}.
\begin{proof}[Proof of \Cref{theorem:hypergraph-cut-sparsifier-combined}]
    We apply \Cref{theorem:black-box-reduction}, with the adversarially-robust self-weighted online algorithm of \Cref{theorem:hypergraph-cut-sparsifier}, and
    an offline algorithm of size $K=\tilde{O}(\eps^{-2} n)$ that succeeds with probability $1-1/\poly(n)$, e.g. \cite{Quanrud2024}.
    Notice that as explained above, sampling probabilities rely only on the values of cuts in the hypergraph so far and hence can be computed using a sparsifier. 
    By \Cref{theorem:hypergraph-cut-sparsifier}, the sparsity of $w'$ is $m'=\tilde{O}(\eps^{-2}n^2 \log m)$.
    Plugging this into \Cref{theorem:black-box-reduction}, we obtain the desired bound.
    Finally, the algorithm succeeds with probability $1-2^{-n}-1/\poly(n)=1-1/\poly(n)$.
\end{proof}

%% file: importance_wth_adaptive_sensitivities.tex
In this section, we prove \Cref{lemma:sampling-game-fixed-probabilities}, showing that for $a=O(\epsilon^{-2}\log\tfrac{\log \Delta}{\epsilon\delta})$, $\sampler$ wins the game against $\adversary$ with probability $1-\delta$.
We will use the following definition and results concerning \emph{martingales}.
\begin{definition}[Martingale]
    A martingale is a sequence $X_0,X_1,\ldots$ of random variables with finite mean, such that for every $i\geq 0$,
    \[
    \E [X_{i+1}|X_i,\ldots,X_0]=X_i.
    \]
\end{definition}
We use Freedman's inequality~\cite{Freedman75}, which is an analogous version of Bernstein's inequality for martingales.
Specifically, we use the following formulation, based on~\cite{Tropp11_matrix_freedman}.
\begin{theorem}[Freedman's Inequality]
    \label{theorem:freedman-inequality}
    Let $X_0,X_1,\ldots,X_n$ be a martingale with $X_0=0$.
    Suppose there exists 
    $M>0,\sigma^2>0$ such that,
    for every $1\leq i\leq n$, %
    $|X_i-X_{i-1}|\leq M$ with probability $1$ (a.s.),
    and the predictable quadratic variation satisfies
    \[\sum_{j=1}^i \var(X_j|X_{j-1},\ldots,X_0) \equiv \sum_{j=1}^i \E[(X_j-X_{j-1})^2|X_{j-1},\ldots,X_0]\leq \sigma^2\] with probability $1$.
    Then, for every $\lambda>0$,
    \[
    \Pr(\max_{i\in[n]}|X_i|>\lambda)  \leq 2\exp\big(-\frac{\lambda^2/2}{\sigma^2 + M\lambda/3}\big).
    \]
\end{theorem}
We are now ready to prove \Cref{lemma:sampling-game-fixed-probabilities}.
\begin{proof}[Proof of \Cref{lemma:sampling-game-fixed-probabilities}]
    By Yao's principle, we can assume without loss of generality that $\adversary$ is deterministic.
    That is, if there was a randomized adversary with randomness $r$ that wins the game with probability $>\delta$, then there must be a choice for $r$ for which the adversary wins with probability $>\delta$. 
    Fixing $r$ to this choice yields a deterministic adversary.
    Furthermore, note that we can assume that $x_1 =1$ without loss of generality by rescaling.

    Let $m$ be the length of the stream.
    For every integer $0\leq t\leq m$, let $X_t=\sum_{i=1}^t \tilde{x}_i-x_i$.
    We have $X_0=0$ and $X_t=X_{t-1}+\tilde{x}_t-x_t$ for $t\geq 1$,
    hence, $\E [X_{t}|X_{t-1},\ldots,X_0]=X_{t-1}$ and thus $X_0,X_1,\ldots$ is a martingale.
    The difference sequence satisfies 
    \[|X_t-X_{t-1}|=|\tilde{x}_t-x_t|\leq  \tfrac{1}{a}\sum_{i=1}^t x_i %
    \]
    and the variance satisfies
    \[
    \var (X_t | X_{t-1},\ldots,X_0) = \tfrac{x_t^2}{p_t}-x_t^2\leq  \tfrac{x_t}{a}\sum_{i=1}^t x_i, %
    \]
    and thus the quadratic variation is $\sum_{i=1}^t \var (X_t | X_{t-1},\ldots,X_0)\leq\tfrac{1}{a}(\sum_{i=1}^t x_i)^2$.
    
    We cannot use Freedman's inequality ``as is'', because $\sum_{i=1}^t x_i$ is a random variable.
    Instead, for the sake of analysis, we consider $L=O(\tfrac{1}{\epsilon}\log \Delta)$ stopped processes, as follows.
    For every $\ell\in [L]$, let $\tau_\ell$ be the first time $t$ for which $\sum_{i=1}^t x_i \geq (1+\epsilon)^\ell$.
    Since $\adversary$ is deterministic, for every $t\leq m$, $x_t$ is determined by $X_0,\ldots,X_{t-1}$, hence it also determines 
    the decision whether $t=\tau_\ell$ (i.e., $\tau_\ell$ is a stopping time).
    We define $Y_{t,\ell}$ as the following random process:
    as long as $t\leq \tau_{\ell}-1$, let $Y_{t,\ell}=X_t$. 
    At $t=\tau_{\ell}$, let the residue be $R = (1+\epsilon)^\ell - \sum_{i=1}^{\tau_{\ell}-1} x_i$, and consider a virtual adversary, that inserts $x_{\tau_\ell,\ell}=R$ and $p_{\tau_\ell,\ell} = \min\{a \tfrac{R}{\sum_{i=1}^{\tau_\ell-1} + R},1\}$. To simplify notations, denote by $\tilde{R}$ the response of $\sampler$.
    Set $Y_{\tau_\ell,\ell}=X_{\tau_\ell-1}+\tilde{R}-R$, and for every $t>\tau_\ell$, $Y_{t,\ell}=Y_{t-1,\ell}$.

    These random processes $Y_{t,\ell}$ are clearly still martingales, and their difference sequence and variance admit the following bounds.
    For $t< \tau_\ell$,
    the difference sequence satisfies
    \[
    |Y_{t,\ell}-Y_{t-1,\ell}|\leq \tfrac{1}{a}\sum_{i=1}^t x_i \leq \tfrac{(1+\epsilon)^\ell}{a},
    \]
    the variance satisfies
    \[
    \var(Y_{t,\ell}|Y_{t-1,\ell},\ldots,Y_{0,\ell})\leq \tfrac{x_t}{a}\sum_{i=1}^t x_i,
    \]
    and hence, $\sum_{i=1}^t \var(Y_{t,\ell}|Y_{t-1,\ell},\ldots,Y_{0,\ell})\leq \tfrac{(1+\epsilon)^{2\ell}}{a}$.
    These same bounds hold for $t=\tau_\ell$, and immediately also for $t>\tau_\ell$.
    By Freedman's inequality (\Cref{theorem:freedman-inequality}),
    \[
        \Pr[\max_{t\in [m]}|Y_{t,\ell}|>\epsilon (1+\epsilon)^\ell] \leq
        2\exp\Big(-\frac{\epsilon^2 (1+\epsilon)^{2\ell}/2}{\tfrac{1}{a} (1+\epsilon)^{2\ell} + \tfrac{\epsilon}{3a}(1+\epsilon)^{2\ell}}\Big)
        \leq 2\exp\big(-\frac{\epsilon^2a }{3}\big).
    \]
    For suitable $a=O(\epsilon^{-2}\log\tfrac{\log \Delta}{\epsilon\delta})$, the probability above is bounded by $\tfrac{\delta}{L}$.
    By a union bound, with probability at least $1-\delta$, we have $\max_{t\in [m]}|Y_{t,\ell}|\leq \epsilon (1+\epsilon)^\ell$ for all $\ell\in [L]$.

    In conclusion,
    for every $t\leq m$, we must have $\sum_{i=1}^t x_i\leq x_1 \Delta \leq \Delta$, where the last inequality is by our assumption that $x_1=1$, hence there exists $\ell\in [L]$ such that $\sum_{i=1}^t x_i\in [(1+\epsilon)^{\ell-1},(1+\epsilon)^{\ell}]$.
    Therefore, $X_t=Y_{t,\ell}$, and we have
    \[
    |X_t|\leq \max_{t\in[m]}|Y_{t,\ell}|\leq \epsilon (1+\epsilon)^\ell\leq \epsilon(1+\epsilon)\sum_{i=1}^t x_i.
    \]
    Rescaling $\epsilon$ concludes the proof.
\end{proof}

%% file: hypergraph-sparsification.tex
This section proves \Cref{theorem:hypergraph-cut-sparsifier}.
It is similar to the construction of cut sparsifiers for graphs using online sampling provided in \cite{AG09}.

We begin by presenting several important definitions, which are based on the work of \cite{Quanrud2024,KPS24}.
Let $H=(V,E)$ be an unweighted hypergraph.
For every partition $V_1,\ldots,V_k$ of $V$, let $E[V_1,\ldots,V_k]$ denote the set of hyperedges that are not entirely contained in any of the $V_i$'s.
The structural properties of hypergraphs which allow us to bound the size of the sparsifier rely on the notion of \emph{normalized cuts}.
For every $k\in [2,|V|]$, a \emph{$k$-cut} in $H$ is a partition of the vertex set $V$ into $k$ disjoint sets $V_1,\ldots,V_k$.
The value of the cut is the number of hyperedges that intersect the cut, denoted by $\mintcut_H(V_1,\ldots, V_k)\eqdef|E[V_1,\ldots,V_k]|$.
Finally, the \emph{normalized cut value} of a $k$-cut is defined as $|E[V_1,\ldots,V_k]|/(k-1)$, we denote the minimum normalized cut value of $H$ by $\lambda(H)$.

For every vertex subset $W\subseteq V$, let $H[W]$ be the sub-hypergraph of $H$ induced by $W$, i.e. the hypergraph on the vertices $W$ that includes only hyperedges $e\in E$ such that $e\subseteq W$.
The \emph{strength} of a hyperedge $e\in E$ is given by 
\begin{equation*}
    \kappa_e^H
    = \max_{W\subseteq V} \lambda(H[W\cup e])
    ,
\end{equation*}
where we remove the superscript $H$ when it is clear from context.
We will also need the following fact.
\begin{fact}
    \label{fact:number-k-cuts}
    Let $n$ be an integer.
    Summing over all $k\in [2,n]$, the 
    number of $k$-cuts in a hypergraph on $n$ vertices is the bell number $B_n$, which in turn is bounded by (Theorem 3.1 from \cite{BT10}),
    \begin{equation*}
        B_n 
        <
        \left( \frac{0.792 n}{\log(n+1)} \right)^n
        \le 2^{n \cdot \log n }
        .
    \end{equation*}
\end{fact}

\subsection{\texorpdfstring{Proof of \Cref{theorem:hypergraph-cut-sparsifier}}{Proof of Theorem}}
Note that we prove the theorem for the stronger notion of $k$-cut sparsifiers, which preserve all $k$-cuts for $k\in [2,n]$ up to multiplicative $(1\pm \epsilon)$ factor.
The algorithm used for constructing the sparsifier is presented in \Cref{algorithm:sample-hypergraph}.
We prove \Cref{theorem:hypergraph-cut-sparsifier} by showing that the algorithm returns a small $(1\pm \epsilon)$-cut sparsifier of the hypergraph $H$ with high probability.
The proof of the theorem is split into two parts:
1) Showing that the output of the algorithm is a $(1\pm \epsilon)$ cut sparsifier with high probability, and
2) bounding the number of hyperedges in the resulting sparsifier.

For every $i\in [m]$, let $H_i=(V, E_i = \left\{ e_1,\ldots, e_i \right\})$ be the hypergraph on the first $i$ hyperedges, and let $H_i'=(V,E_i',w')$ be the sparsifier after the $i$-th insertion, note that $H_i'$ is a weighted hypergraph with weight function $w':E_i'\to \R_{>0}$.
\begin{lemma}[Correctness]
    \label{lemma:cut-sparsification-quality}
    For every adaptive adversary and $i\in [m]$, with probability at least $1-2^{-4n}$,
    \Cref{algorithm:sample-hypergraph} outputs a $(1\pm \epsilon)$-cut sparsifier $H'_i$ of $H_i$.
\end{lemma}
\begin{lemma}[Size]
    \label{lemma:cut-sparsification-size}
    The number of hyperedges in the output of \Cref{algorithm:sample-hypergraph} is $O(\epsilon^{-2}n^2 \log m)$ with probability at least $1-2^{-4n}$.
\end{lemma}
\Cref{theorem:hypergraph-cut-sparsifier} follows by a union bound on the two events.
\ifpods
The proof of \Cref{lemma:cut-sparsification-size} is deferred to \Cref{sec:deferred-hypergraph}.
\fi

\begin{algorithm}
    \caption{SAMPLE-HYPERGRAPH}
    \begin{algorithmic}[1]
        \State $H' \gets (V, E'=\emptyset)$
        \State $\rho \gets K_1 \epsilon^{-2} n \log n$ \Comment{where $K_1$ is a large enough constant}
        \While {new edge ${e_i}$}
            \State $coin\gets True$ with probability $p_i = \min\{\rho/\kappa_{e_i}^{H_i'},1\}$, and otherwise $coin\gets False$
            \If{$coin$}
                \State $E' \gets E'\cup \{e_i\}$
                \State $w_{e_i}' \gets \tfrac{1}{p_i}$
            \EndIf
            \State output $coin$ \Comment{may also output $H'$}
        \EndWhile
    \end{algorithmic}
    \label{algorithm:sample-hypergraph}
\end{algorithm}
\begin{proof}[Proof of \Cref{lemma:cut-sparsification-quality}]
    Fix a $k$-cut $(V_1^*,\ldots,V_k^*)$ and consider a hyperedge $e_i$ that intersects the cut.
    Observe that since the cut intersects the hyperedge $e_i$, it separates the $\kappa_{e_i}$-strong component $W$ containing $e$. 
    Let $W_1,\ldots,W_{k'}$ be the partition of $W$ induced by the cut $(V_1^*,\ldots,V_k^*)$, where $k'\leq k$.
    By definition, we have $\kappa_{e_i}\le \mintcut_{H'_i[W]}(W_1,\ldots,W_{k'})/(k'-1) \le \mintcut_{H'_i[W]}(W_1,\ldots,W_k)$ and since expanding the cut to the entire hypergraph $H'_i$ does not decrease the cut value, we have $\kappa_{e_i}\le \mintcut_{H'_i}(V_1^*,\ldots,V_k^*)$.
    Therefore, the sampling probability satisfies $p_{e_i}= \min\{\rho/\kappa_{e_i},1\}\geq \min\{\rho/\mintcut_{H'_i}(V_1^*,\ldots,V_k^*),1\}$.

    This is precisely the setting of \Cref{thm:importance_samp_adv_robust_main}, since the maximum value of each cut is at most $m$ and its minimum value is at least $1$.
    Recalling that  $T=m\leq 2^n,\delta=2^{-5n\log n}$ and setting $\rho=O(\epsilon^{-2}\log \tfrac{\log T}{\epsilon\delta})=O(\epsilon^{-2}n\log n)$, the probability that the cut is preserved is at least $1-2^{-5n \log n}$.
    The proof concludes by applying a union bound over all $2^{n\log n}$ $k$-cuts.
\end{proof}

\ifpods
\else
\input{lemma-4.3}
\fi

%% file: lemma-4.3.tex
We now turn to bound the number of hyperedges in the sparsifier, proving \Cref{lemma:cut-sparsification-size}.
The proof is similar to Theorem 3.2 in \cite{AG09}.
\begin{proof}[Proof of \Cref{lemma:cut-sparsification-size}]
    We begin by proving several useful claims about hyperedge strengths.
    The first claim is an extension of \cite[Lemma 3.1]{BK96}, on the occurence of $\alpha$-strong components, to hypergraphs.
    Recall that a component $A\subseteq V$ is called $\alpha$-strong if every normalized $k$-cut $A_1,\ldots,A_k$ of $A$ satisfies $\mintcut_H(A_1,\ldots,A_k)/(k-1) \ge \alpha$.
    \begin{claim}
        \label{claim:total-weight-hyperedge-strength-lower-bound}
        A hypergraph with total hyperedge weight at least $\alpha\cdot(n-1)$ has an $\alpha$-strong component. 
    \end{claim}
    \begin{proof}
        The proof is by contradiction.
        Let $n$ be the minimum integer for which there exists a counter example, i.e., a weighted hypergraph $G=(V,E,w)$ that has total hyperedge weight at least $\alpha\cdot(n-1)$, but no $\alpha$-strong component.
        In particular, $G$ is not $\alpha$-strong.
        Hence, there exists a $k$-cut, $V_1,\ldots,V_k$ in $G$ with normalized cut value at most $\mintcut_G(V_1,\ldots,V_k)/(k-1)<\alpha$, for some $k\le n$.
        
        Denote $n_i = |V_i|$ and for every vertex set $S\subseteq V$, denote by $E[S]$ the set of hyperedges in the induced hypergraph $G[S]$.
        By the minimality of $n$, the total hyperedge weight in $G[V_i]$ is at most $\alpha\cdot(n_i-1)$ for all $i\in[k]$.
        Therefore, summing the total weight of hyperedges,
        \begin{equation*}
            \mintcut_G(V_1,\ldots,V_k) + \sum_{i=1}^k w(E[V_i])
            < 
            \alpha(k-1) + \sum_{i=1}^k \alpha(n_i-1)
            = \alpha (n-1)
            ,
        \end{equation*}
        which is in contradiction to the total weight of the hyperedges in $G$.
        Therefore, no such counter example exists.
    \end{proof}
    Next we prove the following useful claim bounding the total weight of hyperedges in the sparsifier.
    \begin{claim}
        \label{lemma:total-hyperedge-weight}
        If $H_i'$ is a $(1\pm \epsilon)$ cut sparsifier of $H_i$ then $\sum_{e\in E_i'} w_e' \le (1+\epsilon) n/2\cdot|E_i|$.
    \end{claim}
    \begin{proof}
        Observe that 
        \begin{equation*}
            \sum_{v\in V} \mintcut_{H_i}(\{v\},V\setminus \{v\})
            \le n\cdot |E_i|
            ,
        \end{equation*}
        since every hyperedge is counted at most $n$ times.
        Similarly, we have that 
        \begin{equation*}
            2\cdot \sum_{e\in E_i'} w_e' 
            \le  
            \sum_{v\in V} \mintcut_{H_i'}(\{v\},V\setminus \{v\})
            \le (1+\epsilon) \sum_{v\in V} \mintcut_{H_i}(\{v\},V\setminus \{v\})
            \le (1+\epsilon) n\cdot |E_i|
            ,
        \end{equation*}
        where the first inequality is since every hyperedge is counted at least twice.
    \end{proof}
    Let $F_{\kappa}= \{ e_j\in E'_i \mid j\le i, \kappa_{e_j} \le \kappa\}$, 
    be the set of all sampled hyperedges that had strength at most $\kappa$ in $H_{j-1}'\cup e_j$ when they were added.
    The following claim bounds the total weight of hyperedges in $F_{\kappa}$.
    \begin{claim}
        \label{claim:sum-weight-F-kappa}
        The total weight of hyperedges in $F_{\kappa}$ is at most $n\kappa(1+1/\rho)$.
    \end{claim}
    \begin{proof}
        Let $G_{\kappa} = (V,F_{\kappa})$ be the sub-hypergraph of the sparsifier that comprises of all the hyperedges in $F_{\kappa}$.
        Observe that if $G_{\kappa}$ has no $(\kappa+\kappa/\rho+1)$-strong component then the total weight of hyperedges in $F_{\kappa}$ is at most $n(\kappa+\kappa/\rho)$ by \Cref{claim:total-weight-hyperedge-strength-lower-bound}.
        Therefore, assume towards contradiction that $G_{\kappa}$ has a $(\kappa+\kappa/\rho+1)$-strong component.

        Let $e$ be the first edge that was sampled into $F_{\kappa}$ that is in the $(\kappa+\kappa/\rho+1)$-strong component.
        Notice that since $H$ is an unweighted hypergraph, the weight of $e$ in the sparsifier is at most $p_e^{-1} \le \rho/\kappa$.
        Hence, removing $e$ can decrease the strength of the component by at most $\kappa/\rho$;
        therefore $G_{\kappa}\setminus e$ has a $(\kappa+1)$-strong component in contradiction to $e$ being sampled with strength at most $\kappa$.
    \end{proof}
    
    We now bound the number of hyperedges in the sparsifier.
    Assume that $H_i'$ is a $(1+\epsilon)$ cut sparsifier of $H_i$, this holds for all $i$ with probability at least $1-2^{-4n}$ by \Cref{lemma:cut-sparsification-quality}.
    By \Cref{lemma:total-hyperedge-weight}, $\sum_{e\in E_i'} w_e' \le (1+\epsilon) \cdot n|E_i|/2$.
    Thus, for all $e'\in E'_i, \kappa_{e'}\leq (1+\epsilon) \rho n\cdot|E_i|/2$, since
    the maximum cut in the graph is $\leq \sum_{e\in E_i'} w_e'$ which is also an upper bound on $\kappa_{e'}$.
    Denote this upper bound on $\kappa$ by $\kappa^*$.
    Therefore, the number of hyperedges is bounded by
    \begin{align*}
        |E'|&=\sum_{\kappa=1}^{\kappa^*} |F_\kappa\setminus F_{\kappa-1}| 
        \leq\sum_{\kappa=1}^{\kappa^*} \tfrac{\rho}{\kappa}\sum_{e\in F_\kappa\setminus F_{\kappa-1}} w'_e  
        && \text{since $w'_e=\tfrac{1}{p_e}\geq \tfrac{\kappa_e}{\rho}$}\\
        &=\sum_{\kappa=1}^{\kappa^*} \tfrac{\rho}{\kappa}( w'(F_\kappa)- w'(F_{\kappa-1})) 
        = \sum_{\kappa=1}^{\kappa^*} \tfrac{\rho}{\kappa}w'(F_\kappa)
        - \sum_{\kappa=0}^{\kappa^*-1} \tfrac{\rho}{\kappa+1}w'(F_\kappa)
        \\
        &=\tfrac{\rho}{\kappa^*+1}w'(F_{\kappa^*})+\rho \sum_{\kappa=1}^{\kappa^*} (\tfrac{1}{\kappa}-\tfrac{1}{\kappa+1})w'(F_\kappa)
        && \text{since $F_0=\emptyset$}\\
        &\leq \rho n(1+\tfrac{1}{\rho})+ \rho \sum_{\kappa=1}^{\kappa^*} \tfrac{1}{\kappa+1}n(1+\tfrac{1}{\rho}) 
        && \text{by \Cref{claim:sum-weight-F-kappa}} \\
        &=O(\rho n \log \kappa^*)=O(\epsilon^{-2}n^2 \log m).
    \end{align*}    
    This concludes the proof of \Cref{lemma:cut-sparsification-size}.
\end{proof}

\subsection{Coreset Properties}
\label{sec:hypergraph-coreset-properties}
In this section we show that cut sparsifiers satisfy the properties of coresets that were presented in \Cref{sec:technical-overview}.
The properties are the following, merge, reduce and linear cost function.
We begin by showing that they satisfy the merge and reduce properties.

Let $E_1,E_2$ be two sets of hyperedges over a vertex set $V$, $H_i = (V,E_i)$ and $H_i'$ a quality $(1\pm\epsilon_i)$-cut sparsifier for $H_i$, for some $\epsilon_i\in(0,1)$.
It is easy to see that $\mintcut_{(V,E_1\cup E_2)}(S)=\mintcut_{H_1}(S)+\mintcut_{H_2}(S)$.
Furthermore, $\mintcut_{H_1'}(S)+\mintcut_{H_2'}(S)\in (1\pm \max\set{\epsilon_1,\epsilon_2})\cdot\mintcut_{(V,E_1\cup E_2)}(S)$.
Therefore, the merge property holds.

To prove the reduce property, let $H$ be some hypergraph, $H'$ be a quality $(1\pm \epsilon_1)$ sparsifier of $H$ and $H''$ be a quality $(1\pm \epsilon_2)$ sparsifier of $H'$.
Notice that,
\begin{equation*}
    \forall S\subseteq V,
    \mintcut_{H''}(S)
    \in (1\pm \epsilon_2)\mintcut_{H'}(S)
    \in (1\pm \epsilon_2)(1\pm \epsilon_1)\mintcut_{H}(S)
\end{equation*}

Finally, notice that the cost function $c:\U\times Q\to\R_+$ is simply $c(e,U)=\indic{0<|e\cap U| < |e|}$, and hence the cost function $C$ satisfies the conditions of \Cref{def:coreset}.

%% file: matrix_spectral_net_argument.tex
In this section, we consider $\ell_p$ subspace embedding and matrix spectral approximation, and prove \Cref{thm:application_ellp_subspace}.
Recall that the input is an $n\times d$ matrix 
given as a stream of rows, denoted $a_1,\ldots,a_n\in \R^d$, and the goal is to maintain an $n'\times d$ weighted submatrix 
that approximates some property of $\Mat{A}$.
We define these problems formally for real matrices, but in the streaming setting, we assume their entries are integers bounded by some $\poly(n)$, as explained in the introduction.
\begin{definition}[Matrix Spectral Approximation]
    Let $d,n,n'\in \N, \eps>0$.
    A matrix $\tilde{\Mat{A}}\in \R^{n'\times d}$ is a \emph{$(1+\eps)$-spectral approximation} of a matrix $\Mat{A}\in \R^{n\times d}$ if
    \[
    (1-\eps)\Mat{A}^\top \Mat{A} 
    \preceq \tilde{\Mat{A}}^\top \tilde{\Mat{A}} 
    \preceq (1+\eps)\Mat{A}^\top \Mat{A}.
    \]
\end{definition}
\begin{definition}[$\ell_p$-Subspace Embedding]
    Let $d,n,n'\in \N, \eps>0$.
    A matrix $\tilde{\Mat{A}}\in \R^{n'\times d}$ is a $(1+\eps)$-approximate $\ell_p$-subspace embedding of a matrix $\Mat{A}\in \R^{n\times d}$ if, for all $x\in \R^d$,
    \[
    \|\tilde{\Mat{A}}x\|_p^p\in (1\pm \eps)\|\Mat{A}x\|_p^p.
    \]
\end{definition}
\begin{remark}
    Matrix spectral approximation is the special case of $\ell_2$-subspace embedding.
\end{remark}

We prove \Cref{thm:application_ellp_subspace}, by
providing an adversarially-robust online algorithm for $\ell_p$ subspace embedding for all $p>0$.
The algorithm is presented the rows of the input matrix in an adaptive stream, and stores $O\Big(\eps^{-2}(d\log\tfrac{\kappa^{OL}}{\eps}+\log\log n)\cdot(d\log (n\kappa^{OL}))^{\max(1,p/2)}\Big)$ rows, where $\kappa^{OL}$ is the online condition number of $\Mat{A}$, defined as the ratio between the largest singular value of $\Mat{A}$ and the smallest non-zero singular value across all $\Mat{A}_i$.
The algorithm assumes a bound on $\kappa^{OL}$ known in advance.

The algorithm, given in \Cref{algorithm:row-sampling-subspace-embedding}, is based on online importance-sampling.
After the $i$-th insertion, the algorithm holds a weighted submatrix $\tilde{\Mat{A}}_i$ of $\Mat{A}_i$.
For parameter $\lambda>0$, define the online importance of $a_i\in \spn\{\Mat{A}_{i-1}\}$ as $s_i' \eqdef \max_{x\in \spn({\Mat{A}_i})}\tfrac{|a_i^\top x|^p}{\|\Mat{A}_i x\|_p^p+\lambda \|x\|_p^p}$, where $\spn({\Mat{A}_i})$ is the row-span of ${\Mat{A}_i}$ (and if $a_i\not\in \spn{\Mat{A}_{i-1}}$, then its online importance equals $1$).
Note that the importance has an additional $\lambda\Norm{x}_p^p$ term, this is because the algorithm uses a ``ridge'' version of the importances for technical reason.
(For $p=2$, this is equivalent to online ridge leverage scores~\cite{CMP16}, defined as $\tau_i = a_i^\top (\Mat{A}_i^\top \Mat{A}_i + \lambda I)^{-1} a_i = \max_{x\in \spn({\Mat{A}_i})}\tfrac{|a_i^\top x|^2}{\|\Mat{A}_i x\|_2^2+\lambda\|x\|_2^2}$.)
We defer the setting of $\lambda$, suffice is to say that it is sufficiently small, so estimating $\|\tilde{\Mat{A}}x\|_p^p+\lambda\|x\|_p^p$ for all $x\in \spn(\Mat{A})$ yields an $\ell_p$ subspace embedding.

Our analysis proceeds similarly to \cite{BravermanHMSSZ21,JPW23}, by analyzing the error on a fixed $\eps$-net $Y$ of the unit ball $B(0,1)$.
Consider a net point $x\in Y\cap \spn({\Mat{A}_i})$.
Every new row $a$ has bounded norm, $\|a\|_p^p\le \poly(n)$ since the entries of $\Mat{A}$ are bounded, hence
for all $i\in[n]$, we have $\|\Mat{A}_i x\|_p^p+\lambda\|x\|_p^p \in [\lambda,\poly(n)]$, and hence adversarial robustness can be obtained
via \Cref{thm:importance_samp_adv_robust_main}. 
(We essentially view $\lambda\|x\|_p^p$ as the first item in the stream, hence when the first actual row arrives, it is sampled with probability at least its online importance with respect to $\|\Mat{A}_i x\|_p^p+\lambda\|x\|_p^p$.)
We proceed with a union bound over the net-points, and extend the correctness from net-points to the entire space by standard arguments.

\begin{algorithm}[t]
    \caption{Row sampling for $\ell_p$ subspace embedding}
    \begin{algorithmic}[1]
        \State $\tilde{\Mat{A}}\gets \emptyset$
        \State $\rho \gets K_1 \cdot \eps^{-2} (d\log\tfrac{\kappa^{OL}}{\eps}+\log\log n)$  \Comment{where $K_1$ is a large enough constant} %
        \State $\lambda \gets n^{-\Omega(pd)}$ %
        \While {new row ${a_i}$}
            \If{$a_i\in\spn(\tilde{\Mat{A}})$}
                \State $s'_i \gets \max_{x\in \spn(\tilde{\Mat{A}})} \tfrac{|a_i^\top x|^p}{\|\tilde{\Mat{A}} x\|_p^p + \lambda\|x\|_p^p}$
            \Else
                \State $s'_i\gets 1$ \label{line:not_in_span}
            \EndIf
            \State $coin\gets True$ with probability $p_i = \min\{\rho s'_i ,1\}$, and otherwise $coin\gets False$
            \If{$coin$}
                \State append the row $p_i^{-1}a_i$ to $\tilde{\Mat{A}}$
            \EndIf
            \State output $coin$ \Comment{may also output $\tilde{\Mat{A}}$}
        \EndWhile
    \end{algorithmic}
    \label{algorithm:row-sampling-subspace-embedding}
\end{algorithm}
The following lemmas provide the guarantees of \Cref{algorithm:row-sampling-subspace-embedding}.
\Cref{thm:application_ellp_subspace} follows by a union bound on these two events.
\begin{lemma}[Correctness of \Cref{algorithm:row-sampling-subspace-embedding}]\label{lem:correctness-subspace-embedding}
    For each adaptive adversary, with probability $1-\delta$,
    for all $i\in [n]$,
    \Cref{algorithm:row-sampling-subspace-embedding} outputs a $(1+\eps)$-approximate $\ell_p$-subspace embedding of $\Mat{A}_i$.
\end{lemma}
\begin{lemma}[Size analysis of \Cref{algorithm:row-sampling-subspace-embedding}]\label{lem:size-subspace-embedding}
    The number of rows in the output of \Cref{algorithm:row-sampling-subspace-embedding} is $O\Big(\eps^{-2}(d\log\tfrac{\kappa^{OL}}{\eps}+\log\log n + \log\tfrac{1}{\delta})\cdot(d\log (n\kappa^{OL})^{\max(1,p/2)})\Big)$ with probability $1-\delta$.
\end{lemma}

\subsection{Proof of \Cref{lem:correctness-subspace-embedding} (Correctness)}
Let $x\in \R^d$ and $i\in [n]$.
We aim to show that $\|\tilde{\Mat{A}}_i x\|_p^p\in (1\pm \eps)\|\Mat{A}_i x\|_p^p$.
Assume without loss of generality that $x\in \spn(\Mat{A}_i)$. Otherwise, we can decompose $x=x_\perp+x_\parallel$, where $x_\parallel\in \spn(\Mat{A}_i)$ and $x_\perp$ in the space orthogonal to $\spn(\Mat{A}_i)$.
Notice that $\Mat{A}_i x_\perp = \tilde{\Mat{A}}_i x_\perp =0$ since $\tilde{\Mat{A}}_i$ consists of a weighted subset of the rows of $\Mat{A}_i$, hence we can indeed assume that $x\in \spn(\Mat{A}_i)$.
The following lemma states formally that is suffices to approximate $\|\tilde{\Mat{A}}_i x\|_p^p+\lambda\|x\|_p^p$ up to $(1\pm\epsilon)$ to get an $\ell_p$ subspace embedding.
\begin{claim}\label{claim:ridge_to_eps_approx}
    For $\lambda=n^{-\Omega(pd)}$, if $\|\tilde{\Mat{A}}_i x\|_p^p + \lambda \|x\|_p^p \in (1\pm \eps)(\|\Mat{A}_i x\|_p^p+\lambda \|x\|_p^p)$ then $\|\tilde{\Mat{A}}_i x\|_p^p\in (1\pm 2\eps)\|\Mat{A}_i x\|_p^p$.
\end{claim}
\begin{proof}
    Denote by $\kappa_0$ the smallest non-zero singular value throughout the execution.
    Observe that $\kappa_0\|x\|_2\leq \|\Mat{A}_i x\|_2$.
    For $0< p\leq 2$, by H\"older's inequality, $\|\Mat{A}_i x\|_p\geq \|\Mat{A}_i x\|_2\geq \kappa_0\|x\|_2 \geq d^{\tfrac{1}{p}-\tfrac{1}{2}} \kappa_0 \|x\|_p$. 
    Similarly, for $p>2$, $\|\Mat{A}_i x\|_p\geq n^{\tfrac{1}{p}-\tfrac{1}{2}} \kappa_0 \|x\|_p$.
    Recall that the entries of $\Mat{A}_i$ are bounded by $\poly(n)$, hence $\kappa_0\geq n^{-O(d)}$. 
    Set $\lambda \leq n^{-\Omega((p+1) d)}\leq \tfrac{\kappa_0^p}{n^p}$, and therefore $\lambda\|x\|_p^p\leq \|\Mat{A}_i x\|_p^p$.
    To conclude, if $\Norm{\tilde{\Mat{A}}x}_p^p + \lambda\Norm{x}_p^p
    \in (1\pm \epsilon)\left( \Norm{\Mat{A}x}_p^p + \lambda\Norm{x}_p^p \right) $, then
\begin{align*}
    \Norm{\tilde{\Mat{A}}x}_p^p + \lambda\Norm{x}_p^p &\in (1\pm 2\epsilon)\Norm{\Mat{A}x}_p^p
    + \lambda\Norm{x}_p^p
    .
\end{align*}
Subtracting $\lambda\Norm{x}_2^2$ from both sides we obtain \Cref{claim:ridge_to_eps_approx}.
\end{proof}

To proceed with the proof of \Cref{lem:correctness-subspace-embedding},
we show that \Cref{algorithm:row-sampling-subspace-embedding} satisfies the guarantees of \Cref{thm:importance_samp_adv_robust_main}, and we indeed obtain a $(1+\eps)$-approximation of $\lambda\|x\|_p^p+ \|\Mat{A}x\|_p^p = \lambda\|x\|_p^p+\sum_{j=1}^i |a_j^\top x|^p$.
Consider $\lambda\|x\|_p^p$ as the first item in the stream, sampled with probability $1$.
At the end of the stream, $\lambda\|x\|_p^p+\sum_{j=1}^n |a_j^\top x|^p\leq \|x\|_p^p \cdot \poly(n^p)$, hence the boundedness requirement is satisfied with $\Delta = \tfrac{\poly(n^p)}{\lambda}$.
The online importance of the $j$-th item is $\tfrac{|a_j^\top x|^p}{\|\tilde{\Mat{A}}_j x\|_p^p+\lambda \|x\|_p^p}\leq s'_j$. 
By \Cref{thm:importance_samp_adv_robust_main} with 
suitable $\delta'=\delta\cdot O(\tfrac{\eps}{\kappa^{OL}})^d$ and $\rho=O(\eps^{-2}\log \tfrac{\log (\Delta/\lambda)}{\eps\delta'})=O(\eps^{-2} (d\log \tfrac{\kappa^{OL}}{\eps}+ \log (p\log n) + \log\tfrac{1}{\delta}))$, we get a $(1+\eps)$-estimate of $\|\Mat{A}_i x\|_p^p + \lambda \|x\|_p^p$ with probability at least $1-\delta'$.

We now extend this to $(1+\eps)$-estimates for a suitable $\eps'$-net, and then extend to all of $\R^d$.
Consider an $\eps'$-net $Y$ of the $\ell_p$ unit ball $B_p(0,1)$ with $\eps'=\tfrac{\eps}{\kappa^{OL}}$.
By standard arguments, the net size is $|Y|\leq O(\tfrac{\kappa^{OL}}{\eps})^d=\tfrac{\delta}{\delta'}$, and a union bound yields correctness for all net-points with probability $1-\delta$.
Let $i\in [n]$ and $x\in \R^d$. As mentioned above, we can assume without loss of generality that $x\in \spn(\Mat{A}_i)$.
We shall represent it as an infinite sum $x=\sum_{j=0}^\infty x_j$, where each $x_j$ is a scalar multiplication of a net-point and $\|x_{j+1}\|_p\leq \eps' \|x_j\|_p$.
Let $y_0\in Y$ be the nearest net-point to $\tfrac{x}{\|x\|_p}$, and denote $x_0=\|x\|_p\cdot y_0$ and 
$r_1 = x-x_0$.
Recursively set $y_j\in Y$ as the nearest net-point to $\tfrac{r_{j}}{\|r_j\|_p}$, and denote $x_j = \|r_j\|_p\cdot y_j$ and $r_{j+1}=x-\sum_{j'=0}^j x_{j'}$.
By definition, $\|\tfrac{r_{j}}{\|r_j\|_p}-y_j\|_p\leq \eps'$, and thus $\|r_{j+1}\|_p\equiv\|r_j-x_j\|_p\leq \eps' \|r_j\|_p$.
We now show that $\|\tilde{\Mat{A}}_ix\|_p\leq (1+\eps)\|\Mat{A}_ix\|_p$.
Denote by $\sigma_1$ the largest singular value of $\Mat{A}$, by standard argument, it is larger than the largest singular value of $\Mat{A}_i$.
Observe,
\[
\sum_{j=1}^\infty \|\Mat{A}_ix_j\|_p\leq \sigma_1 \sum_{j=1}^\infty \|x_j\|_p
\leq \sigma_1 \sum_{j=1}^\infty (\eps')^j \|x_0\|_p
= O(\eps' \sigma_1  \|x_0\|_p)
\leq O(\eps\|\Mat{A}_ix_0\|_p),
\]
and by triangle inequality,
\[
\|\Mat{A}_ix_0\|_p\leq \|\Mat{A}_ix\|_p+\sum_{j=1}^\infty \|\Mat{A}_i x_j\|_p = \|\Mat{A}_ix\|_p+O(\eps \|\Mat{A}_ix_0\|_p).
\]
Thus, $\|\Mat{A}_ix_0\|_p\leq (1+O(\eps))\|\Mat{A}_ix\|_p$.
Therefore,
\[
\|\tilde{\Mat{A}}_ix\|_p\leq \sum_{j=0}^\infty \|\tilde{\Mat{A}}_ix_j\|_p\leq (1+\eps) \sum_{j=0}^\infty \|{\Mat{A}_i}x_j\|_p \leq (1+O(\eps))\|{\Mat{A}}x_0\|_p \leq (1+O(\eps))\|\Mat{A}_ix\|_p.
\]
The other direction that $\|\tilde{\Mat{A}}_ix\|_p\geq (1-O(\eps))\|\Mat{A}_ix\|_p$ is by similar arguments.
Rescaling $\eps$ concludes the proof of \Cref{lem:correctness-subspace-embedding}.

\subsection{Proof of \Cref{lem:size-subspace-embedding} (Size)}

To prove \Cref{lem:size-subspace-embedding}, we need the following result.

\begin{lemma}[Corollary 3.2 of \cite{WoodruffYasuda23}]\label{lem:sum_of_ellp_sensitivities}
    Let a matrix $\Mat{A}\in \R^{n\times d}$ with online condition number $\kappa^{OL}$ and $p\in (0,\infty)$. 
    Define $s_i\coloneqq \max_{x\in \spn(\Mat{A}_i)} \tfrac{|a_i^\top x|^p}{\|\Mat{A}_i x\|_p^p}$.
    Then, $\sum_{i=1}^n s_i = O(d \log (n\kappa^{OL}))^{\max(1,p/2)}$.
\end{lemma}

\begin{proof}[Proof of \Cref{lem:size-subspace-embedding}]
    Denote $S=\sum_i  s'_i$ and $\tilde{S} = \sum_i  \tilde{s}'_i$, where $\tilde{s}'_i$ is $s'_i/p_i$ with probability $p_i$ and $0$ otherwise.

    Since $1\leq \tfrac{\rho s'_i}{p_i}$, we have that the number of sampled rows is $ \leq\rho \tilde{S}$.
    Therefore, to bound the number of sampled rows, it suffices to bound $\tilde{S}$.
    We bound $\tilde{S}$ by another application of \Cref{thm:importance_samp_adv_robust_main}.

    Observe that $s'_1 = 1$ by \Cref{line:not_in_span} of \Cref{algorithm:row-sampling-subspace-embedding}, and in general, $s'_i\leq 1$, hence $S\leq n$. 
    Moreover, we have $\tfrac{s'_i}{\sum_{j=1}^i s'_j}\leq s'_i$, hence $p_i \ge \min\set{\rho\tfrac{s'_i}{\sum_{j=1}^i s'_j},1}$, so \Cref{algorithm:row-sampling-subspace-embedding} performs online importance sampling with respect to $S$, and by \Cref{thm:importance_samp_adv_robust_main}, $\tilde{S}\leq 2 S$ with probability $1-\delta$.
    By \Cref{lem:correctness-subspace-embedding},
    \begin{align*}
        s'_i \equiv \max_{x\in \spn({\tilde{\Mat{A}}_i})} 
        \tfrac{|a_i^\top x|^p}
        {\|{\tilde{\Mat{A}}_i} x\|_p^p + \lambda\|x\|_p^p} 
        &\leq \max_{x\in \spn({\tilde{\Mat{A}}_i})} 
        \tfrac{|a_i^\top x|^p}
        {\tfrac{1}{2}\|{{\Mat{A}}_i} x\|_p^p + \lambda\|x\|_p^p} 
        \\
        &\leq 2\max_{x\in \spn({\Mat{A}_i})} \tfrac{|a_i^\top x|^p}{\|{\Mat{A}_i} x\|_p^p } 
        = 2 s_i,
    \end{align*}
    where we used that $\|{\tilde{\Mat{A}}_i} x\|_p^p \in (1\pm \epsilon) \|\Mat{A}_i x\|_p^p$ and $\epsilon<1$.
    Thus, by \Cref{lem:sum_of_ellp_sensitivities}, $S=O(d \log (n\kappa^{OL}))^{\max(1,p/2)}$, and hence the number of sampled rows is $O(\rho \cdot S) = O(\eps^{-2} (d\log \tfrac{\kappa^{OL}}{\eps} + \log \log n + \log\tfrac{1}{\delta})\cdot (d \log (n\kappa^{OL}))^{\max(1,p/2)})$.
\end{proof}

\subsection{Coreset Properties}
\label{sec:subspace-coreset-properties}
In this section we show that $\ell_p$ subspace embeddings satisfy the properties of coresets that were presented in \Cref{sec:technical-overview}.
The properties are the following, merge, reduce and linear cost function.
The linearity of the cost is immediate, since for matrix $A\in \R^{n\times d}$ and every $x\in \R^d$, we have
$\|Ax\|_p^p = \sum_{i=1}^n |a_i^\top x|^p$.
We now prove the merge and reduce properties.

We begin with the merge property.
Let $A_1,A_2$ be two real matrices with $d$ columns, and let $A'_1,A'_2$ be $(1+\eps_1)$- and $(1+\eps_2)$-approximate $\ell_p$ subspace embeddings of $A_1,A_2$, respectively, for some $\eps_1,\eps_2\in (0,1)$.
It is easy to see that for all $x\in\R^d$, 
\[\left\|\begin{bmatrix}
A_1\\
A_2
\end{bmatrix}
x\right\|_p^p=\|A_1 x\|_p^p+\|A_2 x\|_p^p,\]
and thus 
\[
\left\|\begin{bmatrix}
A_1'\\
A_2'
\end{bmatrix} x\right\|_p^p \in (1\pm \max\set{\epsilon_1,\epsilon_2})\cdot
\left\|\begin{bmatrix}
A_1\\
A_2
\end{bmatrix}
x\right\|_p^p
\]
Therefore, the merge property holds.

To prove the reduce property, let $A$ be some matrix, $A'$ be a $(1+ \epsilon_1)$ sparsifier of $A$ and $A''$ be a  $(1+ \epsilon_2)$ sparsifier of $A'$.
For all $x\in\R^d$, 
\[
\|A'' x\|_p^p\in (1\pm \epsilon_2)\|A'x\|_p^p \in (1\pm \epsilon_2)(1\pm \epsilon_1)\|Ax\|_p^p,
\]
concluding the proof.

\subsection{Applying merge-and-reduce with the online sampling}\label{sec:subspace_merge_reduce}
In this section, we prove \Cref{thm:main_subspace_streaming}. It follows by combining \Cref{thm:application_ellp_subspace} with merge-and-reduce, as shown in \Cref{sec:technical_merge_reduce}.

\begin{proof}[Proof of \Cref{thm:main_subspace_streaming}]
    We apply the black-box wrapper, \Cref{theorem:black-box-reduction}, with the adversarially-robust self-weighted online algorithm of \Cref{thm:application_ellp_subspace}, and
    an offline algorithm that stores $K=O\big(\eps^{-2} d^{\max(1,p/2)}\cdot(\log^2 d \cdot \log \tfrac{d}{\epsilon} + \log \tfrac{1}{\delta})\big)$ rows and succeeds with probability $1-\delta$~\cite{CohenP15,MuscoMWYasuda22,WoodruffYasuda23}.  
    \Cref{thm:application_ellp_subspace} produces a virtual stream with 
    $m'=O(\eps^{-2}(d\log\tfrac{\kappa^{OL}}{\eps} + \log\log n + \log\tfrac{1}{\delta}) \cdot(d\log (n\kappa^{OL})^{\max(1,p/2)})$ rows.
    Plugging this into \Cref{theorem:black-box-reduction}, we obtain $O\big(\eps^{-2}d^{\max(1,p/2)}\cdot\log^3 m'\cdot(\log^2 d \cdot \log \tfrac{d}{\epsilon} + \log \tfrac{m'}{\delta})\big) = \tO(\eps^{-2}d^{\max(1,p/2)} (\log\log(n\kappa^{OL}))^4)$ rows, concluding the proof.
\end{proof}

%% file: merge_and_reduce_robust_proof.tex
\section{Adversarial robustness of merge and reduce}\label{sec:merge_reduce_proof}
For completeness, we include a proof that the merge-and-reduce framework (\Cref{thm:merge_reduce_robust}) is adversarially robust. (\cite{BravermanHMSSZ21} make this claim, but they only provide a proof for a special case.)
\begin{proof}[Proof of \Cref{thm:merge_reduce_robust}]
    The algorithm is the classical merge-and-reduce. The stream is partitioned into blocks of size $b=g(\tfrac{\epsilon}{2\log m},\tfrac{\delta}{m})$. We implicitly construct a full binary tree of depth $\log(m/b)$. Every leaf corresponds to one block and outputs it without processing. Every inner node gets as input its two children outputs $P_1,P_2$, and using fresh randomness, it outputs $\A_{\eps',\delta'}(P_1\cup P_2)$, where $\eps'=\tfrac{\epsilon}{2\log m}$ and $\delta'=\tfrac{\delta}{m}$.
    This tree is maintained implicitly during the stream. The algorithm maintains two blocks at the bottommost tree level, and at most one coreset for the other tree levels. When a node receives its input (i.e., if it's a leaf, then it receives $b$ elements, and if it's an inner node, then it receives its two children outputs), it computes its output and frees the memory of its children (for non-leafs).
    We claim that the union of all stored sets is a $(1+\eps)$-coreset of the input, as desired (throughout, with probability $1-\delta$). The size bound is immediate, hence we focus on correctness.

    First, consider the leaf's level. When a leaf gets a stream $X=\{x_1,x_2,\ldots,x_i\}$, for $i\leq b$, it outputs $X$. This is a deterministic algorithm (and thus adversarially-robust), and $X$ is clearly a coreset of itself.
    Next, consider a non-leaf node. When the node gets its input $P_1, P_2$, it computes $\A_{\eps',\delta'}(P_1\cup P_2)$ using fresh randomness. Notice that the input $P_1, P_2$ is oblivious to the algorithm the node uses, hence the node outputs a $(1+\eps')$-coreset of $P_1\cup P_2$ with probability $1-\delta'$.
    There are less than $m$ nodes in the tree, hence by a union bound, all nodes outputs a $(1+\eps')$-coreset of their respective inputs with probability $1-m\cdot\delta'=1-\delta$. Assume this event happens.

    Now, we prove by induction that the output of every level $i$ node is a $(1+\eps')^i$-coreset of its descendant leafs. At level $0$ (the leafs), the output equals the input, and is clearly a coreset with $\eps''=0$.
    For the inductive step, assume that the output of every level $i$ node is a $(1+\eps')^i$-coreset. Consider a level $i+1$ node, whose input is $P_1,P_2$. By the merge property, $P_1\cup P_2$ is a $(1+\eps')^i$-coreset of the descendant leafs. The node outputs a $(1+\eps')$-coreset of $P_1\cup P_2$, and by the reduce property, this output is a $(1+\eps')^{i+1}$-coreset of the descendant leafs, which concludes the induction.

    Finally, observe that by the merge property, the union of all stored sets is a $(1+\eps')^{\log(m/b)}$-coreset of the input.
    We have that 
    \[(1+\eps')^{\log(m/b)}\leq 1+2\eps'\cdot\log(m/b)\leq 1+\eps,\]
    which concludes the proof.
\end{proof}

%% file: main.bbl
\newcommand{\etalchar}[1]{$^{#1}$}
\begin{thebibliography}{MMWY22}

\bibitem[ABD{\etalchar{+}}21]{AlonBDMNY21}
Noga Alon, Omri Ben{-}Eliezer, Yuval Dagan, Shay Moran, Moni Naor, and Eylon
  Yogev.
\newblock Adversarial laws of large numbers and optimal regret in online
  classification.
\newblock In {\em {STOC} '21: 53rd Annual {ACM} {SIGACT} Symposium on Theory of
  Computing}, pages 447--455, 2021.
\newblock \href {https://doi.org/10.1145/3406325.3451041}
  {\path{doi:10.1145/3406325.3451041}}.

\bibitem[ABJ{\etalchar{+}}22]{ABJSSWZ22}
Mikl{\'{o}}s Ajtai, Vladimir Braverman, T.~S. Jayram, Sandeep Silwal, Alec Sun,
  David~P. Woodruff, and Samson Zhou.
\newblock The white-box adversarial data stream model.
\newblock In {\em {PODS} '22: International Conference on Management of Data},
  pages 15--27. {ACM}, 2022.
\newblock \href {https://doi.org/10.1145/3517804.3526228}
  {\path{doi:10.1145/3517804.3526228}}.

\bibitem[AC24]{AhmadianEdith24}
Sara Ahmadian and Edith Cohen.
\newblock Unmasking vulnerabilities: cardinality sketches under adaptive
  inputs.
\newblock In {\em Proceedings of the 41st International Conference on Machine
  Learning}, ICML. JMLR.org, 2024.

\bibitem[ACGS23]{AssadiCGS23}
Sepehr Assadi, Amit Chakrabarti, Prantar Ghosh, and Manuel Stoeckl.
\newblock Coloring in graph streams via deterministic and adversarially robust
  algorithms.
\newblock In {\em Proceedings of the 42nd {ACM} {SIGMOD-SIGACT-SIGAI} Symposium
  on Principles of Database Systems, {PODS}}, pages 141--153, 2023.
\newblock \href {https://doi.org/10.1145/3584372.3588681}
  {\path{doi:10.1145/3584372.3588681}}.

\bibitem[ACS25]{ahmadian2025costcompressiontightquadratic}
Sara Ahmadian, Edith Cohen, and Uri Stemmer.
\newblock The cost of compression: Tight quadratic black-box attacks on
  sketches for $\ell_2$ norm estimation, 2025.
\newblock To appear in NeurIPS '25.
\newblock \href {https://arxiv.org/abs/2507.16345} {\path{arXiv:2507.16345}}.

\bibitem[ACSS24]{AttiasCSS24}
Idan Attias, Edith Cohen, Moshe Shechner, and Uri Stemmer.
\newblock A framework for adversarial streaming via differential privacy and
  difference estimators.
\newblock {\em Algorithmica}, 86(11):3339--3394, 2024.
\newblock \href {https://doi.org/10.1007/S00453-024-01259-8}
  {\path{doi:10.1007/S00453-024-01259-8}}.

\bibitem[AG09]{AG09}
Kook~Jin Ahn and Sudipto Guha.
\newblock Graph sparsification in the semi-streaming model.
\newblock In {\em Automata, Languages and Programming, 36th Internatilonal
  Colloquium, {ICALP} 2009}, volume 5556 of {\em Lecture Notes in Computer
  Science}, pages 328--338. Springer, 2009.
\newblock \href {https://doi.org/10.1007/978-3-642-02930-1\_27}
  {\path{doi:10.1007/978-3-642-02930-1\_27}}.

\bibitem[BDKS16]{BallardDKS16}
Grey Ballard, Alex Druinsky, Nicholas Knight, and Oded Schwartz.
\newblock Hypergraph partitioning for sparse matrix-matrix multiplication.
\newblock {\em {ACM} Trans. Parallel Comput.}, 3(3):18:1--18:34, 2016.
\newblock \href {https://doi.org/10.1145/3015144} {\path{doi:10.1145/3015144}}.

\bibitem[BDM{\etalchar{+}}20]{BDM+MUWZ20}
Vladimir Braverman, Petros Drineas, Cameron Musco, Christopher Musco, Jalaj
  Upadhyay, David~P. Woodruff, and Samson Zhou.
\newblock Near optimal linear algebra in the online and sliding window models.
\newblock In {\em 61st {IEEE} Annual Symposium on Foundations of Computer
  Science, {FOCS}}, pages 517--528, 2020.
\newblock \href {https://doi.org/10.1109/FOCS46700.2020.00055}
  {\path{doi:10.1109/FOCS46700.2020.00055}}.

\bibitem[BEO22]{BEO22}
Omri Ben{-}Eliezer, Talya Eden, and Krzysztof Onak.
\newblock Adversarially robust streaming via dense-sparse trade-offs.
\newblock In {\em 5th Symposium on Simplicity in Algorithms, {SOSA}}, pages
  214--227. {SIAM}, 2022.
\newblock \href {https://doi.org/10.1137/1.9781611977066.15}
  {\path{doi:10.1137/1.9781611977066.15}}.

\bibitem[BHM{\etalchar{+}}21]{BravermanHMSSZ21}
Vladimir Braverman, Avinatan Hassidim, Yossi Matias, Mariano Schain, Sandeep
  Silwal, and Samson Zhou.
\newblock Adversarial robustness of streaming algorithms through importance
  sampling.
\newblock In {\em Advances in Neural Information Processing Systems,
  {NeurIPS}}, pages 3544--3557, 2021.
\newblock URL:
  \url{https://proceedings.neurips.cc/paper/2021/hash/1d01bd2e16f57892f0954902899f0692-Abstract.html}.

\bibitem[BJWY22]{BJWY22}
Omri Ben{-}Eliezer, Rajesh Jayaram, David~P. Woodruff, and Eylon Yogev.
\newblock A framework for adversarially robust streaming algorithms.
\newblock {\em J. {ACM}}, 69(2):17:1--17:33, 2022.
\newblock \href {https://doi.org/10.1145/3498334} {\path{doi:10.1145/3498334}}.

\bibitem[BK96]{BK96}
Andr{\'{a}}s~A. Bencz{\'{u}}r and David~R. Karger.
\newblock Approximating \emph{s-t} minimum cuts in
  \emph{{\~{O}}}(\emph{n}\({}^{\mbox{2}}\)) time.
\newblock In {\em Proceedings of the Twenty-Eighth Annual {ACM} Symposium on
  the Theory of Computing}, pages 47--55. {ACM}, 1996.
\newblock \href {https://doi.org/10.1145/237814.237827}
  {\path{doi:10.1145/237814.237827}}.

\bibitem[BST19]{BST19}
Nikhil Bansal, Ola Svensson, and Luca Trevisan.
\newblock New notions and constructions of sparsification for graphs and
  hypergraphs.
\newblock In David Zuckerman, editor, {\em 60th {IEEE} Annual Symposium on
  Foundations of Computer Science, {FOCS} 2019}, pages 910--928. {IEEE}
  Computer Society, 2019.
\newblock \href {https://doi.org/10.1109/FOCS.2019.00059}
  {\path{doi:10.1109/FOCS.2019.00059}}.

\bibitem[BT10]{BT10}
Daniel Berend and Tamir Tassa.
\newblock Improved bounds on bell numbers and on moments of sums of random
  variables.
\newblock {\em Probability and Mathematical Statistics}, 30(2):185--205, 2010.

\bibitem[BY20]{Ben-EliezerY20}
Omri Ben{-}Eliezer and Eylon Yogev.
\newblock The adversarial robustness of sampling.
\newblock In {\em Proceedings of the 39th {ACM} {SIGMOD-SIGACT-SIGAI} Symposium
  on Principles of Database Systems, {PODS}}, pages 49--62, 2020.
\newblock \href {https://doi.org/10.1145/3375395.3387643}
  {\path{doi:10.1145/3375395.3387643}}.

\bibitem[CGS22]{CGS22}
Amit Chakrabarti, Prantar Ghosh, and Manuel Stoeckl.
\newblock Adversarially robust coloring for graph streams.
\newblock In {\em 13th Innovations in Theoretical Computer Science Conference,
  {ITCS}}, volume 215 of {\em LIPIcs}, pages 37:1--37:23. Schloss Dagstuhl -
  Leibniz-Zentrum f{\"{u}}r Informatik, 2022.
\newblock \href {https://doi.org/10.4230/LIPICS.ITCS.2022.37}
  {\path{doi:10.4230/LIPICS.ITCS.2022.37}}.

\bibitem[CKN21]{CKN21}
Yu~Chen, Sanjeev Khanna, and Ansh Nagda.
\newblock Sublinear time hypergraph sparsification via cut and edge sampling
  queries.
\newblock In {\em 48th International Colloquium on Automata, Languages, and
  Programming, {ICALP} 2021}, volume 198 of {\em LIPIcs}, pages 53:1--53:21,
  2021.
\newblock \href {https://doi.org/10.4230/LIPICS.ICALP.2021.53}
  {\path{doi:10.4230/LIPICS.ICALP.2021.53}}.

\bibitem[CLN{\etalchar{+}}22]{EdithLyuNelsonSSS22}
Edith Cohen, Xin Lyu, Jelani Nelson, Tamas Sarlos, Moshe Shechner, and Uri
  Stemmer.
\newblock On the robustness of {C}ount{S}ketch to adaptive inputs.
\newblock In {\em Proceedings of the 39th International Conference on Machine
  Learning}, volume 162 of {\em Proceedings of Machine Learning Research},
  pages 4112--4140. PMLR, 2022.
\newblock URL: \url{https://proceedings.mlr.press/v162/cohen22a.html}.

\bibitem[CMP16]{CMP16}
Michael~B. Cohen, Cameron Musco, and Jakub Pachocki.
\newblock Online row sampling.
\newblock In {\em Approximation, Randomization, and Combinatorial Optimization.
  Algorithms and Techniques, {APPROX/RANDOM} 2016}, volume~60 of {\em LIPIcs},
  pages 7:1--7:18, 2016.
\newblock \href {https://doi.org/10.4230/LIPICS.APPROX-RANDOM.2016.7}
  {\path{doi:10.4230/LIPICS.APPROX-RANDOM.2016.7}}.

\bibitem[CNS{\etalchar{+}}25]{cohen2025attackrulealltight}
Edith Cohen, Jelani Nelson, Tamás Sarlós, Mihir Singhal, and Uri Stemmer.
\newblock One attack to rule them all: Tight quadratic bounds for adaptive
  queries on cardinality sketches, 2025.
\newblock To apear in SODA '26.
\newblock \href {https://arxiv.org/abs/2411.06370} {\path{arXiv:2411.06370}}.

\bibitem[CNSS23]{CohenNelsonSarlosStemmer23}
Edith Cohen, Jelani Nelson, Tam\'{a}s Sarl\'{o}s, and Uri Stemmer.
\newblock Tricking the hashing trick: a tight lower bound on the robustness of
  countsketch to adaptive inputs.
\newblock In {\em Proceedings of the Thirty-Seventh AAAI Conference on
  Artificial Intelligence}. AAAI Press, 2023.
\newblock \href {https://doi.org/10.1609/aaai.v37i6.25882}
  {\path{doi:10.1609/aaai.v37i6.25882}}.

\bibitem[CP15]{CohenP15}
Michael~B. Cohen and Richard Peng.
\newblock L\({}_{\mbox{p}}\) row sampling by lewis weights.
\newblock In {\em Proceedings of the Forty-Seventh Annual {ACM} on Symposium on
  Theory of Computing, {STOC}}, pages 183--192, 2015.
\newblock \href {https://doi.org/10.1145/2746539.2746567}
  {\path{doi:10.1145/2746539.2746567}}.

\bibitem[CS24]{ChakrabartiS24}
Amit Chakrabarti and Manuel Stoeckl.
\newblock Finding missing items requires strong forms of randomness.
\newblock In {\em 39th Computational Complexity Conference, {CCC}}, volume 300
  of {\em LIPIcs}, pages 28:1--28:20. Schloss Dagstuhl - Leibniz-Zentrum
  f{\"{u}}r Informatik, 2024.
\newblock \href {https://doi.org/10.4230/LIPICS.CCC.2024.28}
  {\path{doi:10.4230/LIPICS.CCC.2024.28}}.

\bibitem[CWXZ25]{CWXZ25}
Vincent Cohen{-}Addad, David~P. Woodruff, Shenghao Xie, and Samson Zhou.
\newblock Nearly space-optimal graph and hypergraph sparsification in
  insertion-only data streams, 2025.
\newblock \href {https://arxiv.org/abs/2510.18180} {\path{arXiv:2510.18180}}.

\bibitem[CWZ23]{Cohen-AddadWZ23}
Vincent Cohen{-}Addad, David~P. Woodruff, and Samson Zhou.
\newblock Streaming euclidean k-median and k-means with o(log n) space.
\newblock In {\em 64th {IEEE} Annual Symposium on Foundations of Computer
  Science, {FOCS}}, pages 883--908, 2023.
\newblock \href {https://doi.org/10.1109/FOCS57990.2023.00057}
  {\path{doi:10.1109/FOCS57990.2023.00057}}.

\bibitem[CX18]{CX18}
Chandra Chekuri and Chao Xu.
\newblock Minimum cuts and sparsification in hypergraphs.
\newblock {\em {SIAM} J. Comput.}, 47(6):2118--2156, 2018.
\newblock \href {https://doi.org/10.1137/18M1163865}
  {\path{doi:10.1137/18M1163865}}.

\bibitem[Fre75]{Freedman75}
David~A. Freedman.
\newblock {On Tail Probabilities for Martingales}.
\newblock {\em The Annals of Probability}, 3(1):100 -- 118, 1975.
\newblock \href {https://doi.org/10.1214/aop/1176996452}
  {\path{doi:10.1214/aop/1176996452}}.

\bibitem[GLW{\etalchar{+}}24]{GribelyukLWYZ24}
Elena Gribelyuk, Honghao Lin, David~P. Woodruff, Huacheng Yu, and Samson Zhou.
\newblock A strong separation for adversarially robust $\ell_0$ estimation for
  linear sketches.
\newblock In {\em IEEE 65th Annual Symposium on Foundations of Computer Science
  (FOCS)}, pages 2318--2343, 2024.
\newblock \href {https://doi.org/10.1109/FOCS61266.2024.00136}
  {\path{doi:10.1109/FOCS61266.2024.00136}}.

\bibitem[GLW{\etalchar{+}}25]{GribelyukLWYZ25}
Elena Gribelyuk, Honghao Lin, David~P. Woodruff, Huacheng Yu, and Samson Zhou.
\newblock Lifting linear sketches: Optimal bounds and adversarial robustness.
\newblock In {\em Proceedings of the 57th Annual ACM Symposium on Theory of
  Computing}, page 395–406. Association for Computing Machinery, 2025.
\newblock \href {https://doi.org/10.1145/3717823.3718227}
  {\path{doi:10.1145/3717823.3718227}}.

\bibitem[GMT15]{GMT15}
Sudipto Guha, Andrew McGregor, and David Tench.
\newblock Vertex and hyperedge connectivity in dynamic graph streams.
\newblock In {\em Proceedings of the 34th {ACM} Symposium on Principles of
  Database Systems, {PODS} 2015}, pages 241--247. {ACM}, 2015.
\newblock \href {https://doi.org/10.1145/2745754.2745763}
  {\path{doi:10.1145/2745754.2745763}}.

\bibitem[HKM{\etalchar{+}}22]{HassidimKMMS22}
Avinatan Hassidim, Haim Kaplan, Yishay Mansour, Yossi Matias, and Uri Stemmer.
\newblock Adversarially robust streaming algorithms via differential privacy.
\newblock {\em J. {ACM}}, 69(6):42:1--42:14, 2022.
\newblock \href {https://doi.org/10.1145/3556972} {\path{doi:10.1145/3556972}}.

\bibitem[HW13]{HW13}
Moritz Hardt and David~P. Woodruff.
\newblock How robust are linear sketches to adaptive inputs?
\newblock In {\em Proceedings of the Forty-Fifth Annual ACM Symposium on Theory
  of Computing}, page 121–130. Association for Computing Machinery, 2013.
\newblock \href {https://doi.org/10.1145/2488608.2488624}
  {\path{doi:10.1145/2488608.2488624}}.

\bibitem[JPW23]{JPW23}
Shunhua Jiang, Binghui Peng, and Omri Weinstein.
\newblock The complexity of dynamic least-squares regression.
\newblock In {\em 64th {IEEE} Annual Symposium on Foundations of Computer
  Science, {FOCS} 2023}, pages 1605--1627. {IEEE}, 2023.
\newblock \href {https://doi.org/10.1109/FOCS57990.2023.00097}
  {\path{doi:10.1109/FOCS57990.2023.00097}}.

\bibitem[KK15]{KK15}
Dmitry Kogan and Robert Krauthgamer.
\newblock Sketching cuts in graphs and hypergraphs.
\newblock In {\em Proceedings of the 2015 Conference on Innovations in
  Theoretical Computer Science, {ITCS} 2015}, pages 367--376. {ACM}, 2015.
\newblock \href {https://doi.org/10.1145/2688073.2688093}
  {\path{doi:10.1145/2688073.2688093}}.

\bibitem[KKTY21]{KKTY21}
Michael Kapralov, Robert Krauthgamer, Jakab Tardos, and Yuichi Yoshida.
\newblock Towards tight bounds for spectral sparsification of hypergraphs.
\newblock In {\em {STOC} '21}, pages 598--611. {ACM}, 2021.
\newblock \href {https://doi.org/10.1145/3406325.3451061}
  {\path{doi:10.1145/3406325.3451061}}.

\bibitem[KLP25]{KLP25}
Sanjeev Khanna, Huan Li, and Aaron Putterman.
\newblock Near-optimal linear sketches and fully-dynamic algorithms for
  hypergraph spectral sparsification.
\newblock In {\em Proceedings of the 57th Annual {ACM} Symposium on Theory of
  Computing, {STOC} 2025}, pages 1190--1200. {ACM}, 2025.
\newblock \href {https://doi.org/10.1145/3717823.3718239}
  {\path{doi:10.1145/3717823.3718239}}.

\bibitem[KMNS21]{KMNS21_separating}
Haim Kaplan, Yishay Mansour, Kobbi Nissim, and Uri Stemmer.
\newblock Separating adaptive streaming from oblivious streaming using the
  bounded storage model.
\newblock In {\em Advances in Cryptology - {CRYPTO} 2021 - 41st Annual
  International Cryptology Conference, {CRYPTO}}, volume 12827 of {\em Lecture
  Notes in Computer Science}, pages 94--121. Springer, 2021.
\newblock \href {https://doi.org/10.1007/978-3-030-84252-9\_4}
  {\path{doi:10.1007/978-3-030-84252-9\_4}}.

\bibitem[KPS24]{KPS24}
Sanjeev Khanna, Aaron Putterman, and Madhu Sudan.
\newblock Near-optimal size linear sketches for hypergraph cut sparsifiers.
\newblock In {\em 65th {IEEE} Annual Symposium on Foundations of Computer
  Science, {FOCS} 2024}, pages 1669--1706. {IEEE}, 2024.
\newblock \href {https://doi.org/10.1109/FOCS61266.2024.00105}
  {\path{doi:10.1109/FOCS61266.2024.00105}}.

\bibitem[KPS25]{KPS25}
Sanjeev Khanna, Aaron Putterman, and Madhu Sudan.
\newblock Near-optimal hypergraph sparsification in insertion-only and
  bounded-deletion streams.
\newblock In {\em 52nd International Colloquium on Automata, Languages, and
  Programming, {ICALP} 2025}, volume 334 of {\em LIPIcs}, pages 108:1--108:11.
  Schloss Dagstuhl - Leibniz-Zentrum f{\"{u}}r Informatik, 2025.
\newblock \href {https://doi.org/10.4230/LIPICS.ICALP.2025.108}
  {\path{doi:10.4230/LIPICS.ICALP.2025.108}}.

\bibitem[LL22]{LL22}
Yi~Li and Mingmou Liu.
\newblock Lower bounds for sparse oblivious subspace embeddings.
\newblock In {\em {PODS} '22: International Conference on Management of Data},
  pages 251--260. {ACM}, 2022.
\newblock \href {https://doi.org/10.1145/3517804.3526224}
  {\path{doi:10.1145/3517804.3526224}}.

\bibitem[MMWY22]{MuscoMWYasuda22}
Cameron Musco, Christopher Musco, David~P. Woodruff, and Taisuke Yasuda.
\newblock Active linear regression for $\ell_p$ norms and beyond.
\newblock In {\em 63rd {IEEE} Annual Symposium on Foundations of Computer
  Science, {FOCS}}, pages 744--753, 2022.
\newblock \href {https://doi.org/10.1109/FOCS54457.2022.00076}
  {\path{doi:10.1109/FOCS54457.2022.00076}}.

\bibitem[MT20]{MartinssonT20}
Per{-}Gunnar Martinsson and Joel~A. Tropp.
\newblock Randomized numerical linear algebra: Foundations and algorithms.
\newblock {\em Acta Numer.}, 29:403--572, 2020.
\newblock \href {https://doi.org/10.1017/S0962492920000021}
  {\path{doi:10.1017/S0962492920000021}}.

\bibitem[Qua24]{Quanrud2024}
Kent Quanrud.
\newblock Quotient sparsification for submodular functions.
\newblock In {\em Proceedings of the 2024 {ACM-SIAM} Symposium on Discrete
  Algorithms, {SODA} 2024}. {SIAM}, 2024.
\newblock \href {https://doi.org/10.1137/1.9781611977912.187}
  {\path{doi:10.1137/1.9781611977912.187}}.

\bibitem[Sto23]{Stoeckl23}
Manuel Stoeckl.
\newblock Streaming algorithms for the missing item finding problem.
\newblock In {\em Proceedings of the 2023 {ACM-SIAM} Symposium on Discrete
  Algorithms, {SODA}}, pages 793--818, 2023.
\newblock \href {https://doi.org/10.1137/1.9781611977554.CH32}
  {\path{doi:10.1137/1.9781611977554.CH32}}.

\bibitem[STY24]{STY24}
Tasuku Soma, Kam~Chuen Tung, and Yuichi Yoshida.
\newblock Online algorithms for spectral hypergraph sparsification.
\newblock In {\em Integer Programming and Combinatorial Optimization - 25th
  International Conference, {IPCO} 2024}, volume 14679 of {\em Lecture Notes in
  Computer Science}, pages 405--417. Springer, 2024.
\newblock \href {https://doi.org/10.1007/978-3-031-59835-7\_30}
  {\path{doi:10.1007/978-3-031-59835-7\_30}}.

\bibitem[Tro11]{Tropp11_matrix_freedman}
Joel Tropp.
\newblock {Freedman's inequality for matrix martingales}.
\newblock {\em Electronic Communications in Probability}, 16(none):262 -- 270,
  2011.
\newblock \href {https://doi.org/10.1214/ECP.v16-1624}
  {\path{doi:10.1214/ECP.v16-1624}}.

\bibitem[Woo14]{Woodruff14}
David~P. Woodruff.
\newblock Sketching as a tool for numerical linear algebra.
\newblock {\em Found. Trends Theor. Comput. Sci.}, 10(1-2):1--157, 2014.
\newblock \href {https://doi.org/10.1561/0400000060}
  {\path{doi:10.1561/0400000060}}.

\bibitem[WY23]{WoodruffYasuda23}
David~P. Woodruff and Taisuke Yasuda.
\newblock Online lewis weight sampling.
\newblock In {\em Proceedings of the 2023 {ACM-SIAM} Symposium on Discrete
  Algorithms, {SODA}}, pages 4622--4666, 2023.
\newblock \href {https://doi.org/10.1137/1.9781611977554.CH175}
  {\path{doi:10.1137/1.9781611977554.CH175}}.

\bibitem[WZ21]{WoodruffZ21}
David~P. Woodruff and Samson Zhou.
\newblock Tight bounds for adversarially robust streams and sliding windows via
  difference estimators.
\newblock In {\em 62nd {IEEE} Annual Symposium on Foundations of Computer
  Science, {FOCS}}, pages 1183--1196, 2021.
\newblock \href {https://doi.org/10.1109/FOCS52979.2021.00116}
  {\path{doi:10.1109/FOCS52979.2021.00116}}.

\bibitem[WZ24]{WoodruffZ24}
David~P. Woodruff and Samson Zhou.
\newblock Adversarially robust dense-sparse tradeoffs via heavy-hitters.
\newblock In {\em Advances in Neural Information Processing Systems {NeurIPS}},
  2024.
\newblock URL:
  \url{http://papers.nips.cc/paper\_files/paper/2024/hash/14c00f4bc19a5498982b16647998e894-Abstract-Conference.html}.

\bibitem[YNY{\etalchar{+}}19]{YadatiNYNLT19}
Naganand Yadati, Madhav Nimishakavi, Prateek Yadav, Vikram Nitin, Anand Louis,
  and Partha~P. Talukdar.
\newblock Hypergcn: {A} new method for training graph convolutional networks on
  hypergraphs.
\newblock In {\em Advances in Neural Information Processing Systems, NeurIPS},
  pages 1509--1520, 2019.
\newblock URL:
  \url{https://proceedings.neurips.cc/paper/2019/hash/1efa39bcaec6f3900149160693694536-Abstract.html}.

\bibitem[ZHS06]{ZhouHS06}
Dengyong Zhou, Jiayuan Huang, and Bernhard Sch{\"{o}}lkopf.
\newblock Learning with hypergraphs: Clustering, classification, and embedding.
\newblock In {\em Advances in Neural Information Processing Systems
  ({NeurIPS})}, pages 1601--1608. {MIT} Press, 2006.
\newblock URL:
  \url{https://proceedings.neurips.cc/paper/2006/hash/dff8e9c2ac33381546d96deea9922999-Abstract.html}.

\end{thebibliography}
